\documentclass[bj]{imsart}

\RequirePackage[OT1]{fontenc}
\RequirePackage{amsthm,amsmath}
\usepackage[numbers]{natbib}
\RequirePackage[colorlinks,citecolor=blue,urlcolor=blue]{hyperref}

\arxiv{1609.04967}

\usepackage{color}
\usepackage{amssymb}
\usepackage{amsfonts}
\usepackage{comment}
\usepackage[small,bf,figurewithin=section]{caption}
\usepackage{graphicx}
\usepackage{graphics}
\usepackage{nicefrac}
\usepackage{multicol}
\usepackage[shortlabels]{enumitem}
\usepackage{dsfont}
\usepackage{bbm}
\usepackage{url}
\usepackage{capt-of}
\usepackage[english]{babel}
\usepackage[]{inputenc}
\usepackage{float}
\usepackage[caption=false]{subfig}
\usepackage{accents}
\newlength{\dtildeheight}

\startlocaldefs
\numberwithin{equation}{section}
\theoremstyle{plain}
\newtheorem{theorem}{Theorem}[section]
\theoremstyle{remark}
\newtheorem{remark}{Remark}[section]

\newtheorem{lemma}[theorem]{Lemma}
\newtheorem{proposition}[theorem]{Proposition}
\newtheorem{definition}[theorem]{Definition}
\newtheorem{corollary}[theorem]{Corollary}

\newtheorem{condition}[theorem]{Condition}
\newtheorem{fig}[theorem]{Figure}
\newtheorem{example}[theorem]{Example}
\endlocaldefs

\newcommand{\trans}{{}^{^{\intercal}}}

\newcommand{\bthe}{\begin{theorem}}
\newcommand{\ethe}{\end{theorem}}

\newcommand{\ben}{\begin{enumerate}}
\newcommand{\een}{\end{enumerate}}

\newcommand{\beq}{\begin{equation}}
\newcommand{\eeq}{\end{equation}}

\newcommand{\ble}{\begin{lemma}}
\newcommand{\ele}{\end{lemma}}

\newcommand{\bde}{\begin{definition}}
\newcommand{\ede}{\end{definition}}

\newcommand{\bco}{\begin{corollary}}
\newcommand{\eco}{\end{corollary}}

\newcommand{\bpr}{\begin{proposition}}
\newcommand{\epr}{\end{proposition}}

\newcommand{\bproof}{\begin{proof}}
\newcommand{\eproof}{\end{proof}}

\newcommand{\bexam}{\begin{example}\rm}
\newcommand{\eexam}{\halmos\end{example}}

\newcommand{\brem}{\begin{remark}\rm}
\newcommand{\erem}{\halmos\end{remark}}

\newcommand{\bfi}{\begin{fig}}
\newcommand{\efi}{\end{fig}}

\newcommand{\btab}{\begin{tab}}
\newcommand{\etab}{\end{tab}}

\newcommand{\beao}{\begin{eqnarray*}}
\newcommand{\eeao}{\end{eqnarray*}\noindent}

\newcommand{\beam}{\begin{eqnarray}}
\newcommand{\eeam}{\end{eqnarray}\noindent}

\newcommand{\barr}{\begin{array}}
\newcommand{\earr}{\end{array}}

\newcommand{\bdis}{\begin{displaymath}}
\newcommand{\edis}{\end{displaymath}\noindent}
\newcommand{\bs}{\boldsymbol}

\newcommand{\bbn}{\mathbb{N}}
\newcommand{\bbz}{\mathbb{Z}}
\newcommand{\bbr}{\mathbb{R}}
\newcommand{\N}{\mathbb{N}}

\newcommand{\R}{\mathbb{R}}
\newcommand{\E}{\mathbb{E}}

\def\cals{{\mathcal{S}}}

\def\calu{{\mathcal{U}}}
\def\calo{{\mathcal{O}}}
\def\calv{{\mathcal{V}}}

\newcommand{\std}{\stackrel{d}{\rightarrow}}

\newcommand{\stv}{\stackrel{v}{\rightarrow}}

\newcommand{\nto}{n\to\infty}
\newcommand{\kto}{k\to\infty}

\newcommand{\al}{{\alpha}}

\newcommand{\ga}{{\gamma}}

\newcommand{\eps}{\epsilon}
\newcommand{\var}{\mathbb{V}{\rm ar}}
\newcommand{\cov}{\mathbb{C}{\rm ov}}

\newcommand{\eqqcolon}{=\mathrel{\mathop:}}

\newcommand{\ov}{\overline}
\newcommand{\wh}{\widehat}
\newcommand{\wt}{\widetilde}

\DeclareMathOperator*{\diag}{diag}

\DeclareMathOperator*{\argmin}{arg\,min}

\newcommand{\iso}{{\rm iso}}

\newcommand{\halmos}{\quad\hfill\mbox{$\Box$}}  

\allowdisplaybreaks[4]

\begin{document}

\begin{frontmatter}
\title{Semiparametric estimation for isotropic max-stable space-time processes}
\runtitle{Semiparametric estimation for isotropic max-stable space-time processes}

\begin{aug}
\author{\fnms{Sven} \snm{Buhl}\thanksref{a,e1}\ead[label=e1,mark]{sven.buhl@tum.de}}
\author{\fnms{Richard A.} \snm{Davis}\thanksref{b,e2}\ead[label=e2,mark]{rdavis@stat.columbia.edu}%
\ead[label=u2,url]{http://www.stat.columbia.edu/\textasciitilde rdavis}}
\author{\fnms{Claudia} \snm{Kl{\"u}ppelberg}\thanksref{a,e3}\ead[label=e3,mark]{cklu@tum.de}}
\and
\author{\fnms{Christina} \snm{Steinkohl}\thanksref{a,e4}%
\ead[label=e4,mark]{christina.steinkohl@gmail.com}%
\ead[label=u1,url]{http://www.statistics.ma.tum.de}}

\address[a]{Center for Mathematical Sciences and TUM Institute of Advanced Study, Technische Universit{\"a}t M{\"u}nchen, Boltzmannstr. 3, 85748 Garching, Germany.
\printead{e1,e3,e4},
\printead{u1}}

\address[b]{Department of Statistics, Columbia University, 1255 Amsterdam Avenue, New York, NY 10027, USA.
\printead{e2},
\printead{u2}}

\runauthor{S. Buhl et al.}


\end{aug}


\end{frontmatter}


\noindent
Regularly varying space-time processes have proved  useful to study extremal dependence in space-time data.
We propose a semiparametric estimation procedure based on a closed form expression of the extremogram  to estimate parametric models of extremal dependence functions.
We establish the asymptotic properties of the resulting parameter estimates and propose subsampling procedures to obtain asymptotically correct confidence intervals.
A simulation study shows that the proposed procedure works well for moderate sample sizes and is robust to small departures from the underlying model.  
Finally, we apply this estimation procedure to fitting a max-stable process to radar rainfall measurements in a region in Florida. Complementary results and some proofs of key results are presented together with the simulation study in the supplement~\citet{BDKSsupp}.

\section{Introduction}\label{Introduction}

Regularly varying processes provide a useful framework for modeling extremal dependence in continuous time or space. They have been investigated in \citet{HL2005,HL}.
A prominent class of examples consists of max-stable processes. A key example in this paper is the max-stable Brown-Resnick process which was introduced in a time series framework in \citet{Brown}, in a spatial setting in \citet{Schlather2}, and extended to a space-time setting in \citet{Steinkohl}.

In the literature, various dependence models and 
estimation procedures have been proposed for extremal data. 
For the Brown-Resnick process with parametrized dependence structure, inference has been based on composite likelihood methods. 
In particular, pairwise likelihood estimation has been found useful  to estimate parameters in a max-stable process. 
A description of this method can be found in \citet{Ribatet} for the spatial setting, and \citet{Huser} in a space-time setting. 
Asymptotic results for  pairwise likelihood estimates and detailed analyses in the space-time setting for the model analysed in this paper are given in \citet{Steinkohl2}.
Unfortunately, parameter estimation using composite likelihood methods can be laborious, since the computation and subsequent optimization of the objective function is time-consuming. 
Also the choice of good initial values for the optimization of the composite likelihood is essential.

In this paper we introduce a new semiparametric estimation procedure for regularly varying processes which is based on the extremogram as a natural extremal analog of the correlation function for stationary processes. 
The extremogram was introduced in \citet{Davis2}  for time series (also in \citet{Fasen}), and they show consistency
and asymptotic normality of an empirical extremogram estimate under weak mixing conditions.
The empirical extremogram and its asymptotic properties in a spatial setting have been investigated in  \citet{buhl3} and \citet{Cho}.
It can serve as a useful graphical tool for assessing extremal dependence structures in spatial and space-time processes that provides clues about potential parametric models, a critical step in the model building paradigm.  For example, compatibility with various assumptions such as isotropy and stationarity (see \citet{buhl1} and \citet{Steinkohl2} for some examples), can be assessed by examining invariance of the empirical extremogram when computed over specially chosen subsets of the data.  Ultimately, a number of families of proposed parametric models are often fitted before deciding on a particular class of models.  Therefore it is of interest to be able to not only have a procedure that can compute estimates rapidly, but also to serve as a check on the efficacy of model choice. Additionally, the new estimation procedure  allows one to provide parameter estimates that can be used as initial values in more refined procedures, such as composite likelihood.

Our semiparametric estimation method assumes a spatially isotropic and additively separable dependence structure for regularly varying space-time processes.
We first estimate the extremogram nonparametrically by its empirical version, where we can hence separate space and time. 
Weighted linear regression is then applied in order to produce parameter estimates. 
Asymptotic normality of these semiparametric estimates requires asymptotic normality of the empirical extremogram, and we apply the CLT with mixing conditions as provided in \cite{buhl3}. The rate of convergence can be improved by a bias correction term, a fact which we explain in detail.  
The proofs of the asymptotic properties of semiparametric spatial and temporal parameter estimates are analogous, and we present the details on the spatial parameters only, referring to \citet{buhlphd}, Chapter~3, for details about the asymptotic properties of the semiparametric temporal parameter. 

In a second step we establish asymptotic normality of the weighted least squares parameter estimates. When the dependence parameters have bounded support, as for the Brown-Resnick process in Section~\ref{sec7:Fr}, constrained optimization has to be applied. Then also the limit law differs depending whether the true parameters lie on the boundary or not.
Since the asymptotic covariance matrix in the normal limit is difficult to access, we apply subsampling procedures to obtain pointwise confidence intervals for the parameters.

The semiparametric estimates converge at a slower rate than the square root rate of a fully parametric procedure such as pairwise likelihood estimation. 
However, it is known that likelihood-based estimates may be inefficient and even not consistent if the model is slightly misspecified.  
The semiparametric estimates, however, are often unaffected by slight deviations in the model.  This is proved in Section~9 and illustrated in Section~10 of the supplement~\cite{BDKSsupp}, where data are generated from a Brown-Resnick process, but with observational noise.  The semiparametric estimates clearly outperform pairwise likelihood estimates in this case.  On the other hand, the semiparametric estimates perform admirably well relative to the pairwise likelihood estimates when the underlying process is in fact a Brown-Resnick process. 

Our paper is organized as follows.
Section~\ref{Desmodel} defines regularly varying processes in space and time and their extremogram. Based on gridded data, the nonparametric extremogram estimation is derived and used for parametric model fitting.
Asymptotic normality of the parameter estimates is established in Section~\ref{Asymptotics}.
Section~\ref{spacesection1} is dedicated to the asymptotic normality of the empirical extremogram; and
Section~\ref{spacesection2} deals with the asymptotic properties of the parameter estimates.
The subsampling procedure -- as well as results and proofs for our setting -- is given in Section~7 of the supplement~\cite{BDKSsupp}.
In Section~\ref{sec7:Fr} we apply the semiparametric method to the Brown-Resnick process and verify the required conditions. Here we also calculate the bias corrected estimator. 
{We test our new semiparametric estimation procedure in a simulation study presented in the supplement~\cite{BDKSsupp}
and compare it to pairwise likelihood estimation, both when applied to data generated by a Brown-Resnick process and when the data are affected by observational noise.  In the latter, our  procedure produces estimates with less bias than those based on pairwise likelihood (see Section~10 of the supplement~\cite{BDKSsupp}).
The paper concludes with an analysis of daily rainfall maxima in a region in Florida in Section~\ref{sec:Florida}, where we also compare the semiparametric estimates with previously obtained pairwise likelihood estimates.
The supplement~\cite{BDKSsupp}} contains four sections, on subsampling, on $\al$-mixing of the Brown-Resnick process, a robustness result for the bias corrected estimator, and a simulation study.

\section{Model description and semiparametric estimates}\label{Desmodel}

In this paper we consider strictly stationary \textit{regularly varying processes} in space and time $\{\eta(\bs s,t): \bs s \in \mathbb{R}^{d-1}, t \in [0,\infty)\}$ for $d\in\N$, where all finite-dimensional distributions are regularly varying (cf. \citet{HL} for definitions and results in a general framework and \citet{Resnick3} for details about multivariate regular variation).
Throughout, $f(n)\sim g(n)$ means that $\lim_{\nto}\frac{f(n)}{g(n)} =1$.
As a prerequisite, we define for every finite set $\mathcal{I} \subset \mathbb{R}^{d-1} \times [0,\infty)$ with cardinality $|\mathcal{I}|$  the vector 
$$\eta_{\mathcal{I}}:= (\eta(\bs s,t) : (\bs s,t) \in \mathcal{I})\trans.$$ Let furthermore $\| \cdot \|$ be the Euclidean norm on $\mathbb{R}^{d-1}$.

\begin{definition}[Regularly varying stochastic process] \label{def:reg_var}
A strictly stationary stochastic space-time process $\{\eta(\bs s,t): (\bs s,t) \in \mathbb{R}^{d-1} \times [0,\infty)\}$ is called {\em regularly varying}, if there exists some normalizing sequence $0<a_n \to \infty$ such that $\mathbb{P}(|\eta(\bs 0,0)| > a_n) \sim n^{-d}$ as $n \to \infty$, and if for every finite set $\mathcal{I} \subset \mathbb{R}^{d-1}\times [0,\infty)$,
\begin{align} \label{regvar}
n^d \mathbb{P}\Big(\frac{\eta_{\mathcal{I}}}{a_n} \in \cdot\Big) \stv \mu_{\mathcal{I}}(\cdot), \quad n \rightarrow \infty,
\end{align}
for some non-null Radon measure $\mu_{\mathcal{I}}$ on the Borel sets in $\overline{\mathbb{R}}^{|\mathcal{I}|}\backslash\{\bs 0\}$. In that case, 
$$\mu_{\mathcal{I}}(xC)=x^{-\beta}\mu_{\mathcal{I}}(C), \quad x>0,$$ 
for every Borel set $C$ in $\overline{\mathbb{R}}^{|\mathcal{I}|}\backslash\{\bs 0\}$.
The notation $\stv$ stands for vague convergence, and $\beta>0$ is called the {\em index of regular variation}.
\end{definition}
For every $(\bs s,t) \in \mathbb{R}^{d-1}\times [0,\infty)$ and $\mathcal{I}=\{(\bs s,t)\}$ we set $\mu_{\{(\bs s,t)\}}(\cdot)=\mu_{\{(\bs 0,0)\}}(\cdot)=:\mu(\cdot),$ which is justified by stationarity.
Throughout we furthermore consider the space-time process $\{\eta(\bs s,t): (\bs s,t) \in \mathbb{R}^{d-1} \times [0,\infty)\}$ to be spatially isotropic. Together with the assumption of strict stationarity, this means that extremal dependence between two space-time points $(\bs s_1,t_1)$ and $(\bs s_2,t_2)$ is only driven by the spatial and temporal lags $v:=\|\bs s_1-\bs s_2\|$ and $u:=|t_1-t_2|$, respectively, and we can define the extremogram only as a function of $v$ and $u$. The extremogram was introduced 
for spatial and space-time processes by \citet{buhl3} and \citet{Cho}, based on \citet{steinkohlphd}, and can be regarded as a correlogram for extreme events.

\begin{definition}[The extremogram]
For a regularly varying strictly stationary isotropic space-time process $\{\eta(\bs{s},t): (\bs{s},t)\in\bbr^{d-1}\times [0,\infty)\}$ we define the \emph{space-time extremogram} for two $\mu$-continuous Borel sets $A$ and $B$ in $\ov{\mathbb{R}} \backslash\{0\}$ (i.e. $\mu(\partial A)=\mu(\partial B)=0$) such that $\mu(A)>0$ by
\begin{equation}
\rho_{AB}(v,u) = \lim_{n\to \infty}\frac{\mathbb{P}\left(\eta(\bs{s}_1,t_1)/a_n\in A,\eta(\bs{s}_2,t_2)/a_n\in B\right)}{\mathbb{P}\left(\eta(\bs{s}_1,t_1)/a_n \in A\right)},
\label{extremogramst}
\end{equation}
where $v=\|\bs s_1-\bs s_2\|$ and $u=|t_1-t_2|$.
Setting $A=B=(1,\infty)$, this reduces to the \textit{tail dependence coefficient} $\chi(v,u) = \rho_{(1,\infty)(1,\infty)}(v,u)$. 
\end{definition}

In what follows we propose a two-step semiparametric estimation procedure of a parametric model of the extremogram. In particular, we assume that the model is additively separable such that setting either the temporal lag $u$ or the spatial lag $v$ equal to $0$, it can be linearly parametrized as 
\begin{align}
T_1(\chi(v,0))=T_1(\chi(v,0;C_1,\alpha_1))=C_1+\alpha_1 v, \quad (C_1,\alpha_1) \in \Theta_{\mathcal{S}}, \quad v \geq 0, \label{eq:param_space}
\end{align}
and
\begin{align} 
 T_2(\chi(0,u))=T_2(\chi(0,u;C_2,\alpha_2))=C_2+\alpha_2 u, \quad (C_2,\alpha_2) \in \Theta_{\mathcal{T}} \quad u\geq 0,\label{eq:param_time}
\end{align}
  where $T_1$ and $T_2$ are known suitable strictly monotonous continuously differentiable transformations and the parameters $(C_1,\alpha_1)$ and $(C_2,\alpha_2)$ lie in appropriate parameter spaces $\Theta_{\mathcal{S}}$ and $\Theta_{\mathcal{T}}$. We refer to $(C_1,\alpha_1)$ as the {\em spatial parameter} and to $(C_2,\alpha_2)$ as the {\em temporal parameter}.
Equations~\eqref{eq:param_space} and \eqref{eq:param_time} are the basis for parameter estimates. We replace the extremogram on the left hand side in both of these equations by nonparametric estimates sampled at different lags. 
Then we use constrained weighted least squares estimation in a linear regression framework to obtain parameter estimates. 

For better understanding, we stick to the $2$-dimensional spatial case $d-1=2$; however, the method can directly be generalized and applied to higher dimensions.
The  estimation procedure is based on the following observation scheme for the space-time data. 
\begin{condition}\label{grid}
(1) \, The locations lie on a regular grid
$$\cals_n= \big \{(i_1,i_2): i_1,i_2 \in \left\{1,\ldots,n\right\}\big \} 
= \big \{\bs{s}_i: i=1,\ldots,n^2\big \} .$$
(2) \, The time points are equidistant, given by the set $\{t_1,\ldots,t_T\}$.
\end{condition}

\brem
The assumption of a regular grid can be relaxed in various ways. A simple, but notationally more involved extension is the generalization to rectangular grids, cf. \citet{buhl3}, Section~3. Furthermore, it is possible to assume that the observation area consists of random locations given by points of a Poisson process, see for instance \citet{Cho}, Section~2.3, or \citet{steinkohlphd}, Section~4.5.2.
 Also  deterministic, but irregularly spaced locations, could be considered as treated in \cite{steinkohlphd} in Section~4.5.1 in the context of pairwise likelihood estimation.
{In order to make our method transparent we focus on observations on a regular grid.}
\erem

The following scheme provides the semiparametric estimation procedure in detail. \\
Denote by $\calv$ and $\mathcal{U}$ finite sets of spatial and temporal lags, on which the estimation is based. 
Concerning their choice, we generally include those lags which show clear extremal dependence between locations or time points.
Larger lags should not be considered, since they may introduce a bias in the least squares estimates, similarly as in pairwise likelihood estimation; cf. \citet{buhl1}, Section~5.3.
One way to determine the range of clear extremal dependence are permutation tests, which we describe at the end of Section~\ref{sec:Florida}. \\

\noindent (1) Nonparametric estimates for the extremogram:\\
{Summarize all pairs of $\cals_n$ which give rise to the same spatial lag $v\in\calv$ into 
$$N(v) = \{(i,j)\in\{1,\dots,n^2\}^2: \ \|\bs{s}_i-\bs{s}_j\|=v\}.$$}
For all $t\in\left\{t_1,\ldots,t_T\right\}$ estimate the {\em spatial extremogram} by
\begin{equation}\label{extremogramspace}
\wh{\chi}^{(t)}(v,0) = \frac{\dfrac{1}{|N(v)|}
\sum\limits_{i=1}^{n^2}\sum\limits_{{j=1\atop \|\bs{s}_i-\bs{s}_j\|=v}}^{n^2} \mathds{1}_{\left\{\eta(\bs{s}_i,t)>q,\eta(\bs{s}_j,t)>q\right\}}}
{\dfrac1{n^2}\sum\limits_{i=1}^{n^2} \mathds{1}_{\left\{\eta(\bs{s}_i,t)>q\right\}}}, \quad v \in \calv,
\end{equation}
where 
$q$ is a large quantile (to be specified) of the standard unit Frech\'et distribution. \\[2mm]
For all $\bs{s}\in\cals_n$ estimate the {\em temporal extremogram} by
\begin{equation}\label{extremogramtime}
\wh{\chi}^{(\bs{s})}(0,u) = \frac{\frac{1}{T-u}\sum\limits_{k=1}^{T-u} \mathds{1}_{\left\{\eta(\bs{s},t_k)>q,\eta(\bs{s},t_k+u)>q\right\}}}{\frac{1}{T}\sum\limits_{k=1}^T\mathds{1}_{\left\{\eta(\bs{s},t_k)>q\right\}}}, \quad u\in\mathcal{U}\,,
\end{equation}
where again $q$ is a large (possibly different) quantile of the standard unit Frech\'et distribution. \\[2mm]
\noindent (2) The overall ``spatial'' and ``temporal'' extremogram estimates are defined as averages over the temporal and spatial locations, respectively; i.e.,
\beam
\wh{\chi}(v,0) &=& \frac{1}{T}\sum_{k=1}^T \wh{\chi}^{(t_k)}(v,0), \quad v \in \calv,\label{extrspacemean}\\
\wh{\chi}(0,u) &=& \frac{1}{{n^2}}\sum_{i=1}^{n^2}\wh{\chi}^{(\bs{s}_i)}(0,u), \quad u\in \mathcal{U}. \label{extrtimemean}
\eeam
(3)  Parameter estimates for $C_1,\alpha_1, C_2$ and $\alpha_2$ are found by using 
 weighted least squares estimation:
\begin{align}\label{minspace}
\begin{pmatrix}
\wh{C}_1 \\ 
\wh{\alpha}_1
\end{pmatrix} 
= \argmin_{(C_1,\alpha_1) \in \Theta_{\mathcal{S}}} 
&\sum_{v\in \calv}  w_v  \Big(T_1(\wh{\chi}(v,0))-\big(C_1 + \alpha_1 v\big)\Big)^2,
\end{align} 
\begin{align}\label{mintime}
\begin{pmatrix}\wh{C}_2 \\ 
\wh{\alpha}_2\end{pmatrix}
=\argmin_{(C_2,\alpha_2) \in \Theta_{\mathcal{T}}} 
&\sum_{u\in \mathcal{U}}w_u \Big(T_2(\wh{\chi}(0,u))-\big(C_2 + \alpha_2 u\big)\Big)^2,
\end{align}
with weights $w_u>0$ and $w_v>0$. 

We call the estimates $(\wh{C}_1,\wh{\alpha}_1)$ and $(\wh{C}_2,\wh{\alpha}_2)$ {\em weighted least squares estimates} (WLSE).
{This approach bears similarity with that proposed by \citet{EKS}, who suggest semiparametric weighted least squares estimation of the parameters of parametric models of the stable tail dependence function based on iid random vector observations.}

\section{Asymptotic properties of the WLSE}\label{Asymptotics}

In this section we investigate aymptotic properties of the WLSE $(\wh{C}_1,\wh{\alpha}_1)$ and $(\wh{C}_2,\wh{\alpha}_2)$. Recall from \eqref{minspace} and \eqref{mintime} that they are functions of the averaged empirical extremogram $\wh{\chi}(\cdot,\cdot)$. Its definition is given in~\eqref{extrspacemean} and \eqref{extrtimemean} and implies that we first need CLTs of the pointwise empirical extremograms  $\wh{\chi}^{(t)}$ and $\wh{\chi}^{(\bs{s})}$ for a fixed time point $t$ and a fixed location $\bs{s}$, respectively. 
 Sections~\ref{spacesection1} and \ref{spacesection2} focus on the {spatial parameters}. 
 {The corresponding results for the temporal case can be derived similarly by replacing $n$ with $\sqrt{T}$ and can be found with full details in \citet{buhlphd}, Chapter~3 for the Brown-Resnick space-time process. 
We use several results for the extremogram provided in Section~8 of the supplement~\cite{BDKSsupp} 
and in \citet{buhl3}.

\subsection{Asymptotics of the empirical spatial extremogram}\label{spacesection1}

We show a CLT for the empirical spatial extremogram of regularly varying space-time processes, which is defined in~\eqref{regvar} and based on a finite set of observed spatial lags 
$$\calv= \{v_1,\ldots, v_p\},$$
which show clear extremal dependence as explained in Section~\ref{Desmodel}.
First we state conditions under which the empirical extremogram centred by the pre-asymptotic version is asymptotically normal.

\begin{theorem}\label{extspace2} 
For a fixed time point $t\in \{t_1,\ldots,t_T\}$, consider a regularly varying spatial process $\left\{\eta(\bs{s},t): \bs{s}\in \bbr^2\right\}$ as defined in Definition~\ref{def:reg_var}. Let $a_n$ be a sequence as in~\eqref{regvar}. Assume that there  exists $\ga>0$ that satisfies $\max\{v_1,\ldots,v_p\} \leq \ga$, such that the following conditions are satisfied:
\begin{enumerate}[(M1)]
\item 
$\{\eta(\bs{s},t): \bs{s}\in\bbr^2\}$ is $\alpha$-mixing with $\alpha$-mixing coefficients $\alpha_{k,\ell}(\cdot)$. 
\end{enumerate}
There exist sequences $m=m_n, r=r_\nto $ with $m_n/n \to 0$ and $r_n/m_n \to 0$ as $\nto$ such that the following hold:
\begin{enumerate}[(M1)]
\setcounter{enumi}{1}
\item 
$m_n^{2}r_n^{2}/n \to 0$.
\label{cond0}
\item 
\label{cond1}
For all $\epsilon>0$: 
\begin{align*}
&\lim\limits_{\kto } \limsup\limits_{\nto }
\sum\limits_{\bs h \in \mathbb{Z}^{2} :  k< \Vert{\bs h}\Vert \leq r_n}
 m_n^2 \\
& \quad \quad\quad  \mathbb{P}\Big(\max\limits_{\bs s \in B(\bs 0,\gamma)} 
|\eta(\bs s,t)|>\epsilon a_m, \max\limits_{\bs s' \in B(\bs h,\gamma)} |\eta(\bs s',t)|>\epsilon a_m\Big)=0,
\end{align*} 
where $B(\bs h,\ga):=\{\bs s\in\mathbb{Z}^2 : \|\bs s-\bs h\|\le\ga\}$ {for $\bs h\in\R^2$}.
\item 
\begin{enumerate}[(i)]
\item 
$\lim\limits_{\nto } m_n^2 \sum\limits_{\bs h \in \mathbb{Z}^{2} : \Vert{\bs h}\Vert > r_n} \alpha_{1, 1}(\|\bs h\|)=0$, \label{cond2} 
\item 
$\sum\limits_{\bs h \in \mathbb{Z}^{2}} \alpha_{p,q}(\Vert{\bs h}\Vert) < \infty$ \ \mbox{ for } $2\le p+q\le 4$, \label{cond3}
\item 
$\lim\limits_{\nto }  m_n n \ \alpha_{1,{n^2}}(r_n)=0$, \label{cond4} 
\end{enumerate}
\end{enumerate}

Then the empirical spatial extremogram $\wh{\chi}^{(t)}(v,0)$ defined in \eqref{extremogramspace} with the quantile $q=a_m$ satisfies
\begin{align}
\frac{n}{m_n}\big(\wh{\chi}^{(t)}(v,0) - \chi_n(v,0)\big)_{v \in \calv}\stackrel{d}{\to} \mathcal{N}(\bs 0,\Pi^{(\iso)}_1), \quad \nto, \label{CLT_nonave}
\end{align}
where the covariance matrix $\Pi^{(\iso)}_1$ is specified in equation~\eqref{isomat} below, and $\chi_n$ is the pre-asymptotic spatial extremogram,
\begin{equation}
\chi_n(v,0) = \frac{\mathbb{P}(\eta(\bs{0},0)>{a_m},\eta(\bs{h},0)>a_m)}{\mathbb{P}(\eta(\bs{0},0)>a_m)}, \quad v=\|\bs h\|\in \calv\,.
\label{preasymptotic1}
\end{equation}
\end{theorem}

\bproof 
Theorem~\ref{extspace2} is a direct application of Theorem~{4.2} of \citet{buhl3} to the process $\{\eta(\bs{s},t): \bs{s}\in \bbr^2\}$ for $d=2$ and $A=B=(1,\infty)$. 
For the specification of the asymptotic covariance matrix we need to adapt that theorem {to the isotropic case,} where each spatial lag $v_i$ arises from a set of different  vectors $\bs h$, all with same Euclidean norm $v_i$.
For $i \in \{1,\ldots,p\}$ such that $v_i \in \calv$, we summarize these into
$$L(v_i):=\{\bs h \in \mathbb{Z}^2: \|\bs h\|=v_i\}=\{\bs h_1^{(i)},\ldots,\bs h_{\ell_i}^{(i)}\},$$
where $\ell_i:=|L(v_i)|$. 
We conclude that 
\begin{align}
&\frac{n}{m_n} \big(\wh{\chi}^{(t)}(\bs h_1^{(i)},0)-\chi_n(\bs h_1^{(i)},0),\ldots,\wh{\chi}^{(t)}(\bs h_{\ell_i}^{(i)},0)-\chi_n(\bs h_{\ell_i}^{(i)},0)\big)\trans_{{i=1,\ldots,p}}
 \std \mathcal{N}(\bs 0,\Pi^{(\textnormal{space})}_1), \nonumber
\end{align}
where $\Pi^{(\textnormal{space})}_1$ is specified in equation~{(4.3)-(4.6)} of \cite{buhl3}.
Note the slight misuse of notation committed here for the sake of simplicity: by $\wh{\chi}^{(t)}(\bs h,0)$ (instead of $\wh{\chi}^{(t)}(v,0)$) we denote the empirical extremogram for each single vector $\bs h \in L(v_i)$  specified above; i.e., 
\begin{equation*}
\wh{\chi}^{(t)}(\bs h,0) = \frac{\dfrac{1}{|N(\bs h)|}
\sum\limits_{i=1}^{n^2}\sum\limits_{{j=1\atop \bs{s}_i-\bs{s}_j=\bs h}}^{n^2} \mathds{1}_{\left\{\eta(\bs{s}_i,t)>q,\eta(\bs{s}_j,t)>q\right\}}}
{\dfrac1{n^2}\sum\limits_{i=1}^{n^2} \mathds{1}_{\left\{\eta(\bs{s}_i,t)>q\right\}}},
\end{equation*}
where $N(\bs h):=\{(i,j)\in\{1,\dots,n^2\} : \bs{s}_i-\bs{s}_j=\bs h\}$ (instead of $N(v)$). Analogously we define the pre-asymptotic extremogram $\chi_n(\bs h,0)$ w.r.t. a vector $\bs h$.

It holds that $|N(v_i)|=\sum_{\bs h \in L(v_i)} |N(\bs h)|$.
Isotropy implies furthermore for the pre-asymptotic extremogram that $\chi_n(v_i,0)=\chi_n(\bs h,0)$ for all $\bs h \in L(v_i)$, such that
\begin{align}
\chi_n(v_i,0)=\sum\limits_{\bs h \in L(v_i)} \frac{|N(\bs h)|}{|N(v_i)|} \chi_n(v_i,0)
=\sum\limits_{\bs h \in L(v_i)} \frac{|N(\bs h)|}{|N(v_i)|} \chi_n(\bs h,0)\label{rhpre}
\end{align}
as well as, by the definition of the estimator in \eqref{extremogramspace},
\begin{align}
\wh{\chi}^{(t)}(v_i,0)
= \sum\limits_{\bs h \in L(v_i)} \frac{|N(\bs h)|}{|N(v_i)|} \wh{\chi}^{(t)}(v_i,0)
= \sum\limits_{\bs h \in L(v_i)} \frac{|N(\bs h)|}{|N(v_i)|} \wh{\chi}^{(t)}(\bs h,0). \label{rhhat}
\end{align}
We conclude by \eqref{rhpre} and \eqref{rhhat} that 
\begin{align*}
\wh{\chi}^{(t)}(v_i,0)-\chi_n(v_i,0)
&=\sum\limits_{\bs h \in L(v_i)}  \frac{|N(\bs h)|}{|N(v_i)|} \big(\wh{\chi}^{(t)}(\bs h,0)-\chi_n(\bs h,0) \big).
\end{align*}
To obtain a concise representation of the asymptotic normal law for the isotropic extremogram, we define row vectors $({|N(\bs h)|}/{|N(v_i)|}: \bs h \in L(v_i))$ for $i=1,\dots,p$.
Set $L:=\sum_{i=1}^p \ell_i$ and define the $p\times L-$matrix 
\begin{align}
N:=
\begin{pmatrix} 
\Big(\frac{|N(\bs h)|}{|N(v_1)|}: \bs h \in L(v_1)\Big) & \bs 0 & \bs 0 & \bs 0\\
 \bs 0 & \Big(\frac{|N(\bs h)|}{|N(v_2)|}: \bs h \in L(v_2)\Big) & \bs 0 & \bs 0\\
\vdots & \vdots & \ddots & \bs 0\\
 \bs 0 & \bs 0 & \bs 0 & \Big(\frac{|N(\bs h)|}{|N(v_p)|}: \bs h \in L(v_p)\Big)
\end{pmatrix}. \label{iso_N}
\end{align}
Then we find 
\beao
&&\frac{n}{m_n} \big(\wh{\chi}^{(t)}( v_i,0)-\chi_n( v_i,0)\big)\trans_{i=1,\ldots,p}\\
&=&\frac{n}{m_n} N \,\big(\wh{\chi}^{(t)}(\bs h_1^{(i)},0)-\chi_n(\bs h_1^{(i)},0),\ldots,\wh{\chi}^{(t)}(\bs h_{\ell_i}^{(i)},0)-\chi_n(\bs h_{\ell_i}^{(i)},0)\big)\trans_{i=1,\ldots,p} \\
&\std & \mathcal{N}(\bs 0, N \Pi^{(\textnormal{space})}_1 N\trans), \quad \nto,
\eeao
such that
\begin{align}
\Pi^{(\iso)}_1:=N\Pi^{(\textnormal{space})}_1N\trans. \label{isomat}
\end{align}
\eproof

\begin{corollary}\label{coroll_average}
Under the conditions of Theorem~\ref{extspace2} the averaged spatial extremogram in \eqref{extrspacemean} satisfies (with covariance matrix $\Pi^{(\textnormal{iso})}_2$ specified in \eqref{defPi} below)
\begin{align}
\frac{n}{m_n} \Big(\frac{1}{T}\sum_{k=1}^T\wh{\chi}^{(t_k)}(v,0) - \chi_n(v,0)\Big)_{v \in \calv}\stackrel{d}{\to} \mathcal{N}(\bs 0,\Pi^{(\textnormal{iso})}_2), \quad \nto. \label{CLT_ave}
\end{align}
\end{corollary}

\begin{proof}
For the first part of the proof, we neglect spatial isotropy. This part is similar to the proof of Theorem~{4.2} in \citet{buhl3} and Corollary~3.4 of \citet{Davis2}. We use the notation of the proof of Theorem~\ref{extspace2}. Enumerate the set of spatial lag vectors inherent in the estimation of the extremogram as $\{\bs h_1^{(i)},\ldots,\bs h_{\ell_i}^{(i)}: i=1,\ldots,p\}$ and let $\ga \geq \max\{v_1,\ldots,v_p\}$.
Define the vector process 
\begin{align*}
\{\bs Y(\bs s): \bs s \in \mathbb{R}^2\}=\{(\eta(\bs s + \bs h, t_k): \bs h \in B(\bs 0, \ga))\trans_{k=1,\dots,T}: \bs s \in \mathbb{R}^2\}.
\end{align*}
Let $A=B=(1,\infty)$. 
Consider $i=1,\ldots,p$, $j=1,\ldots,\ell_i$, and $k=1,\ldots,T$. 
Define sets $D^{(i)}_{j,k}$ by 
$$\{\bs Y(\bs{s})\in D^{(i)}_{j,k}\} = \{\eta(\bs{s},t_k)\in A, \eta(\bs{s}',t_k)\in B : \bs{s}-\bs{s}' = \bs h_j^{(i)}\},$$
and the sets $D_{k}$ by
$$\{\bs Y(\bs{s})\in D_{k}\} = \{\eta(\bs{s},t_k)\in A\}.$$ 
For $\bs h \in \mathbb{R}^2$ let $B_T(\bs h,\ga):=B(\bs h,\ga) \times \{t_1,\ldots,t_T\}$. For $\mu_{B_T(\bs 0, \ga)}$-continuous Borel sets $C$ and $D$ in $\ov\R^{T|B(\bs 0,\ga)|}\backslash\{\bs 0\}$, regular variation yields the existence of the limit measures 
\begin{align*}
\mu_{B_T(\bs 0, \ga)}(C) &:=\lim_{\nto} m_n^{2} \mathbb{P}\Big(\frac{\bs Y(\bs 0)}{m_n^{2}} \in C\Big)\\
\tau_{{B_T(\bs 0,\gamma) \times B_T(\bs h,\gamma)}} (C \times D)
&:= \lim_{\nto} m_n^{2}\mathbb{P}\Big(\frac{\bs Y(\bs 0)}{m_n^2} \in C, \frac{\bs Y(\bs h)}{m_n^{2}} \in D\Big).
\end{align*}
By time stationarity we have $\mu_{B_T(\bs 0, \ga)}(D_{k})=\mu(A)$,
\beam\label{aschi}
\wh{\chi}^{(t_k)}(\bs h_j^{(i)},0) \sim \wh{R}_{m_n}(D^{(i)}_{j,k},D_{k}):=\wh{\mu}_{B_T(\bs 0,\ga),m_n}(D^{(i)}_{j,k})/\wh{\mu}_{B_T(\bs 0,\ga),m_n}(D_{k}),\quad \nto,
\eeam
where the $\wh{\mu}_{B_T(\bs 0,\ga),m_n}(\cdot)$ are empirical estimators of $\mu_{B_T(\bs 0, \ga)}(\cdot)$ defined as 
\begin{align}\label{lln}
\widehat{\mu}_{B_T(\bs 0,\ga),m_n}(\cdot) := \Big(\frac{m_n}{n}\Big)^2 \sum_{\bs s\in \cals_n} \mathbbmss{1}_{\{\frac{\bs Y(\bs s)}{m_n^2} \in \cdot\}}.
\end{align} 
Likewise we have for the pre-asymptotic quantities
\beam\label{apre}
\chi_n(\bs h_j^{(i)},0)=R_{m_n}(D^{(i)}_{j,k},D_{k}):=\frac{\mathbb{P}(\bs Y(\bs 0)/m_n^{2} \in D^{(i)}_{j,k})}{\mathbb{P}(\bs Y(\bs 0)/m_n^{2} \in D_{k})}=:\frac{{\mu}_{B_T(\bs 0,\ga),m_n}(D^{(i)}_{j,k})}{{\mu}_{B_T(\bs 0,\ga),m_n}(D_{k})},
\eeam
which are independent of time $t_k$ by stationarity.
For notational ease we abbreviate in the following
$${\mu}_{B_T(\bs 0,\ga)}(\cdot) = {\mu}_{\ga}(\cdot),\quad
{\mu}_{B_T(\bs 0,\ga),m_n}(\cdot) = {\mu}_{\ga,m_n}(\cdot), \quad\mbox{and}\quad
\widehat{\mu}_{B_T(\bs 0,\ga),m_n}(\cdot) = \widehat{\mu}_{\ga,m_n}(\cdot) 
$$
For each $k \in \{1,\ldots,T\}$ we now define the matrices 
$$F^{(k)}=[F_1,F_2^{(k)}]$$
with $F_1 \in \bbr^{L \times L}$ and $F_2^{(k)} \in \bbr^L$ given by
$$F_1=\diag(\mu(A)) \quad\mbox{and}\quad 
F_2^{(k)}:=(-\mu_{\ga}(D^{(1)}_{1,k}),\ldots,-\mu_{\ga}(D^{(1)}_{\ell_1,k}),\ldots,-\mu_{\ga}(D^{(p)}_{\ell_p,k}))^{\top}.$$ 
Although $F_2^{(k)}$ is constant over $k \in \{1,\ldots, T\}$ by time stationarity, we keep the index to clarify the notation. 
Define the $TL \times T(L+1)$-matrix $\bs F$ and the column vector ${\bs{\wh\chi}} - {\bs \chi_n}$ with $TL$ components as
\begin{align*}
{\bs F} :=\begin{pmatrix}
F^{(1)} & \bs 0 & \bs 0 & \bs 0\\ \bs 0 & F^{(2)} & \bs 0 & \bs 0\\ 
\vdots & \vdots & \ddots & \bs 0\\ \bs 0 & \bs 0 & \bs 0 & F^{(T)}
\end{pmatrix} 
\quad\mbox{and}\quad
{\bs{\wh\chi}} - {\bs \chi_n}  := 
\begin{pmatrix}
\wh{\chi}^{(t_1)}(h^{(1)}_1,0)-\chi_n(h^{(1)}_1,0) \\
\vdots \\
\wh{\chi}^{(t_1)}(h^{(1)}_{\ell_1},0)-\chi_n(h^{(1)}_{\ell_1},0) \\
\vdots \\
\wh{\chi}^{(t_1)}(h^{(p)}_{\ell_p},0)-\chi_n(h^{(p)}_{\ell_p},0) \\
\vdots \\
\wh{\chi}^{(t_T)}(h^{(p)}_{\ell_p},0)-\chi_n(h^{(p)}_{\ell_p},0)
\end{pmatrix}. 
\end{align*}
Define the vector $(\bs{\wh R}_{m_n} - \bs{ R}_{m_n} )$ with the quantities from \eqref{aschi} and the corresponding pre-asymptotic quantities from \eqref{apre} exactly in the same way.
Furthermore, define for $k=1,\dots,T$ the vectors in $\R^{L+1}$ 
\begin{align*}
&\bs{\mu}_{\ga,m_n}^{(k)}=\\
&\big(
 \mu_{\ga,m_n}(D_{1,k}^{(1)}),\dots,\mu_{\ga,m_n}(D_{\ell_1,k}^{(1)}),\dots\dots,
 \mu_{\ga,m_n}(D_{1,k}^{(p)}),\dots,\mu_{\ga,m_n}(D_{\ell_p,k}^{(p)}), \mu_{\ga,m_n}(D_k)
 \big)\trans,
\end{align*}
which we stack one on top of the other giving a vector $\bs{\mu}_{\ga,m_n} \in \R^{T(L+1)}$,
and $\bs{\wh{\mu}}_{\ga,m_n} $ analogously. 
Then we obtain
$${\bs{\wh\chi}} - {\bs \chi}_n= (1+o(1)) (\bs{\wh R}_{m_n} - \bs{ R}_{m_n} )= \frac{1+o_p(1)}{\mu(A)^2} \bs F\,  (\bs{\wh{\mu}}_{\ga,m_n} - \bs{\mu}_{\ga,m_n}), \quad \nto,
$$
where the last step follows as in the proof of Theorem~{4.2} of \cite{buhl3} and involves Slutzky's theorem. 
Using ideas of the proof of their Lemma~5.1, we observe that as $\nto$, 
\begin{align*}
\lefteqn{ \cov \Big[\wh{\mu}_{B_T(\bs 0, \ga),m_n}(C),\wh{\mu}_{B_T(\bs 0, \ga),m_n}(D)\Big] }\\[2mm]
& \sim \Big(\frac{m_n}{n}\Big)^2\Big(\mu_{B_T(\bs 0, \ga)}(C \cap D)+\sum_{\bs 0 \neq \bs h \in \mathbb{Z}^2 }\tau_{B_T(\bs{0},\ga)\times B_T(\bs h,\ga)}(C \times D)\Big)
=: \Big(\frac{m_n}{n}\Big)^2 c_{C,D}.
\end{align*}
With $\Sigma \in \bbr^{T(L+1) \times T(L+1)}$ defined as
\begin{align*}
\Sigma =
\begin{pmatrix}
c_{D^{(1)}_{1,1},D^{(1)}_{1,1}}  & \cdots & c_{D^{(1)}_{1,1},D_{1}} & \cdots & c_{D^{(1)}_{1,1},D^{(p)}_{1,T}} & \cdots & c_{D^{(1)}_{1,1},D_{T}} \\
\vdots &  \ddots & \vdots & \ddots & \vdots & \ddots & \vdots \\ c_{D_{T},D^{(1)}_{1,1}} &  \cdots & c_{D_{T},D_{1}} & \cdots & c_{D_{T},D^{(p)}_{1,T}} & \cdots & c_{D_{T},D_T}
\end{pmatrix},
\end{align*}
we thus conclude that 
\begin{align*}
\frac{n}{m_n} \begin{pmatrix}
\wh{\chi}^{(t_1)}(\bs h^{(1)}_1,0)-\chi_n(\bs h^{(1)}_1,0) \\
\vdots \\
\wh{\chi}^{(t_T)}(\bs h^{(p)}_{\ell_p},0)-\chi_n(\bs h^{(p)}_{\ell_p},0)
\end{pmatrix} \std \mathcal{N}(\bs 0,\mu(A)^{-4} \bs F \Sigma (\bs F)^{\top}).
\end{align*}
To obtain the asymptotic covariance matrix in the spatially isotropic case, we proceed as in the proof of Theorem~\ref{extspace2}. 
We define the $Tp\times TL$-matrix 
\begin{align*}
\bs N:=
\begin{pmatrix}
N & \bs 0 & \bs 0 & \bs 0\\
 \bs 0 & N & \bs 0 & \bs 0\\
  \vdots & \vdots & \ddots & \bs 0\\
   \bs 0 & \bs 0 & \bs 0 & N
\end{pmatrix}
\end{align*}
with $N$ given in equation~\eqref{iso_N}. 
Then we have
\begin{align*}
&\frac{n}{m_n} \begin{pmatrix}
\wh{\chi}^{(t_1)}(v_1,0)-\chi_n(v_1,0) \\
\vdots \\
\wh{\chi}^{(t_T)}(v_p,0)-\chi_n(v_p,0)
\end{pmatrix}
=\frac{n}{m_n} \bs N 
\begin{pmatrix}
\wh{\chi}^{(t_1)}(\bs h^{(1)}_1,0)-\chi_n(\bs h^{(1)}_1,0) \\
\vdots \\
\wh{\chi}^{(t_T)}(\bs h^{(p)}_{\ell_p},0)-\chi_n(\bs h^{(p)}_{\ell_p},0)
\end{pmatrix} \\
&\std \quad \mathcal{N}(\bs 0,\mu(A)^{-4} \bs N \bs F \Sigma (\bs N \bs F)^{\top}), \quad \nto,
\end{align*}
and we conclude that for the averaged spatial extremogram  the statement  holds with 
\begin{small}
\begin{align}\label{defPi}
\Pi^{(\iso)}_2 =  &{\mu(A)^{-4}}{T^{-2}} 
\begin{pmatrix}
 				1 & 0 \cdots 0 & 1 & 0 \cdots 0 & \cdots  & 1 & 0 \cdots  0\\
                                0 & 1\cdots 0 & 0 & 1 \cdots 0 & \cdots & 0 & 1 \cdots 0 \\
                                && \ddots & \\
                                0& 0 \cdots 1 & 0 & 0 \cdots 1 & \cdots  & 0 & 0 \cdots 1\end{pmatrix}  \bs N \bs F \Sigma (\bs N \bs F)^{\top}  \nonumber\\
                                & \begin{pmatrix}
                                1 & 0 \cdots 0 & 1 & 0 \cdots 0 & \cdots  & 1 & 0 \cdots  0\\
                                0 & 1\cdots 0 & 0 & 1 \cdots 0 & \cdots & 0 & 1 \cdots 0 \\
                                && \ddots & \\
                                0& 0 \cdots 1 & 0 & 0 \cdots 1 & \cdots  & 0 & 0 \cdots 1
                                \end{pmatrix}^{\top}.
\end{align}
\end{small}
\end{proof}

\begin{condition}\label{remarkextspace}
In the CLTs~\eqref{CLT_nonave} and \eqref{CLT_ave}, the pre-asymptotic extremogram \eqref{preasymptotic1}  can be replaced by the theoretical one (eq.~ \eqref{extremogramst} with $A=B=(1,\infty)$), provided that
\begin{equation}\label{condition5}
\frac{n}{m_n}\left(\chi_n(v,0) - \chi(v,0)\right) \to 0,\quad \nto,
\end{equation}
is satisfied for all spatial {lags} $v\in \calv$. In particular, we then obtain
\begin{align}
\frac{n}{m_n} \big(\wh{\chi}(v,0) - \chi(v,0)\big)_{v \in \calv}\stackrel{d}{\to} \mathcal{N}(\bs 0,\Pi^{(\textnormal{iso})}_2), \quad \nto. \label{CLT_true}
\end{align}

\end{condition}
This bias condition turns out to be central in order to obtain a CLT for the WLSE $(\wh{C}_1,\wh{\alpha}_1)$ in Section~\ref{spacesection2} below. However, even if it is not satisfied, the empirical extremogram keeps its important asymptotic interpretation as a conditional probability of extremal events. Furthermore there are cases where we can resort to a bias correction, ensuring again a CLT for $(\wh{C}_1,\wh{\alpha}_1)$. For examples we refer to Section~\ref{sec7:Fr} below.

\subsection{Asymptotic properties of spatial parameter estimates}\label{spacesection2}

In this section we state conditions that yield asymptotic normality of the {WLSE $(\wh C_1,\wh\alpha_1)$} of Section~\ref{Desmodel}. Recall the weighted least squares optimization problem~\eqref{minspace}; i.e.,
\begin{align*}
\begin{pmatrix}
\wh{C}_1 \\ 
\wh{\alpha}_1
\end{pmatrix} 
= \argmin_{(C_1,\alpha_1) \in \Theta_{\mathcal{S}}} 
&\sum_{v\in \calv}  w_v  \Big(T_1(\wh{\chi}(v,0))-\big(C_1 + \alpha_1 v\big)\Big)^2.
\end{align*} 

To show asymptotic normality of the WLSE, we define the design matrix $X$ and weight matrix $W$ as
$$X = [\bs{1},(v: v \in \calv)\trans]\in\R^{p\times 2} \quad \text{ and } \quad W=\diag\{w_v: v \in \calv\}\in\R^{p\times p},$$
respectively, where $\bs{1}=(1,\ldots,1)\trans \in \mathbb{R}^{p}$.
If neither $C_1$ nor $\alpha_1$ have bounded support, then the WLSE; i.e., the solution to \eqref{minspace}, is given by 
$$\wh{\psi}_1:=\begin{pmatrix}\wh{C}_1 \\ \wh{\alpha}_1
\end{pmatrix}= (X\trans W X)^{-1}X\trans W (T_1(\wh{\chi}(v,0)))\trans_{v \in \calv}.$$
If one of the parameters $C_1$ or $\alpha_1$ does have bounded support, we need to constrain $\wh{\psi}_1$ properly, obtaining a CLT that might differ considerably from that given in Theorem~\ref{asysemispace} below. An important example of this is treated in Section~\ref{sec7:Fr}.

\begin{theorem}\label{asysemispace}
For a fixed time point $t\in \{t_1,\ldots,t_T\}$, consider a regularly varying spatial process $\left\{\eta(\bs{s},t): \bs{s}\in \bbr^2\right\}$ as defined in Definition~\ref{def:reg_var}. Assume that it satisfies the conditions of Theorem~\ref{extspace2}.
Let $\bs{\wh{\psi}}_1 = (\wh{C}_1,\wh{\alpha}_1)\trans$ denote the WLSE resulting from the minimization problem \eqref{minspace} and $\bs{\psi}^{*}_1=(C_1^*,\alpha_1^*)\trans\in \Theta_{\cals}$ the true parameter vector. Assume that the CLT~\eqref{CLT_true} holds, possibly after a bias correction of the empirical extremogram $(\wh{\chi}_v :v \in \calv)$. Then for a suitably chosen scaling sequence $m_n$, we obtain, as $\nto$,
\begin{equation}
\frac{n}{m_n}\Big(\bs{\wh{\psi}}_1 - \bs{\psi}_1^{*}\Big) \stackrel{d}{\to} \mathcal{N}(\bs 0,Q_x^{(w)}G\Pi^{(\textnormal{iso})}_2 G{Q_x^{(w)}}\trans).
\end{equation}
Here $\Pi^{(\textnormal{iso})}_2$ is the covariance matrix given in \eqref{defPi}, 
\begin{align}\label{defG}
Q_x^{(w)} &= (X\trans W X)^{-1}X\trans W  \quad \text{ and} \quad
G = \diag\big \{T_1'(\chi(v,0)): \ v \in \calv\hspace*{0.1cm}\big \},
\end{align}
where $T_1'(x)$ denotes the derivative of $T_1(x)$ with respect to $x$ for $0<x<1$.
\end{theorem}

\begin{proof}
Using the multivariate delta method together with the CLT~\eqref{CLT_true} it directly follows that
$$\frac{n}{m_n}\big(T_1(\wh{\chi}(v,0))-T_1(\chi(v,0))\big)_{v \in \calv} \stackrel{d}{\to} \mathcal{N}(\bs{0},G\Pi^{\textnormal{(iso)}}_2G),\quad  n\to \infty,$$
where $G$ is defined in \eqref{defG}. 
Since 
$$ \min\limits_{(C_1,\alpha_1) \in \Theta_{\cals}}\sum_{v \in \calv}  w_v  \big(T_1(\chi(v,0))-\big(C_1 + \alpha_1 v\big)\big)^2=\sum_{v \in \calv}  w_v  \big(T_1(\chi(v,0))-\big(C_1^* + \alpha_1^* v\big)\big)^2,$$
we find the well-known property of unbiasedness of the WLSE, 
\beao
Q_x^{(w)} (T_1(\chi(v,0)))\trans_{v \in \calv} = \argmin\limits_{(C_1,\alpha_1) \in \Theta_{\cals}}\sum_{v \in \calv}  w_v  \big(T_1(\chi(v,0))-\big(\log(\theta_1) + \alpha_1 x_v\big)\big)^2 
= \bs{\psi}_1^*.
\eeao
It follows that, as $\nto$,
$$\frac{n}{m_n}\left(\bs{\wh{\psi}}_1-\bs{\psi}_1^*\right)= \frac{n}{m_n} Q_x^{(w)} \big(T_1(\wh{\chi}(v,0))-T_1(\chi(v,0))\big)_{v \in \calv} \stackrel{d}{\to} \mathcal{N}\left(\bs{0},Q_x^{(w)}G\Pi^{(\textnormal{iso})}_2 G {Q_x^{(w)}}\trans\right).$$
\end{proof}

\section{{Example: the Brown-Resnick process}}\label{sec7:Fr}

We illustrate the results of the previous sections by applying them to a max-stable strictly stationary and isotropic Brown-Resnick space-time process with representation
\begin{align}\label{limitfield}
\eta(\bs{s},t) = \bigvee\limits_{j=1}^\infty \left\{\xi_j \,  e^{W_j(\bs{s},t)-\delta(\|\bs{s}\|,t)} \right\},\quad (\bs{s},t)\in\R^2\times [0,\infty),
\end{align}
where 
$\{\xi_j : j\in\N\}$ are points of a Poisson process on $[0,\infty)$ with intensity $\xi^{-2}d\xi$ and the dependence function $\delta$ is \textit{nonnegative and conditionally negative definite}; i.e., for every $m \in \mathbb{N}$ and every $(\bs s^{(1)},t^{(1)}),\ldots,(\bs s^{(m)},t^{(m)}) \in \mathbb{R}^2 \times [0,\infty)$, it holds that 
$$\sum_{i=1}^m \sum_{j=1}^m a_i a_j \delta(\|\bs s^{(i)}-\bs s^{(j)}\|,|t^{(i)}-t^{(j)}|) \leq 0$$ 
for all $a_1,\ldots,a_m \in \mathbb{R}$ summing up to 0. The processes
$\{W_j(\bs{s},t): \bs{s} \in \mathbb{R}^2, t \in [0, \infty)\}$ are independent replicates of a Gaussian process $\{W(\bs s,t): \bs{s} \in \mathbb{R}^2, t \in [0, \infty)\}$ with stationary increments,
$W(\bs{0},0)=0$, $\mathbb{E} [W(\bs{s},t)]=0$ and covariance function
\begin{align*}
\cov[W(\bs{s}^{(1)},t^{(1)}),&W(\bs{s}^{(2)},t^{(2)})] 
\\
=& \, \delta(\|\bs{s}^{(1)}\|,t^{(1)})+\delta(\|\bs{s}^{(2)}\|,t^{(2)}) -\delta(\|\bs{s}^{(1)}-\bs{s}^{(2)}\|, |t^{(1)}-t^{(2)}|).
\end{align*}
Representation \eqref{limitfield} goes back to \citet{deHaan}, \citet{Gine} and \citet{Schlather2}. 
All finite-dimensional distributions are multivariate extreme value distributions with standard unit Fr\'echet margins, hence they are in particular multivariate regularly varying. Furthermore, they are characterized by the dependence function $\delta$, which is termed the \textit{semivariogram} of the process $\{W(\bs s, t)\}$ in geostatistics: For $(\bs{s}^{(1)},t^{(1)}), (\bs{s}^{(2)},t^{(2)}) \in \R^2\times[0,\infty)$, it is given by
$$\var[W(\bs{s}^{(1)},t^{(1)})-W(\bs{s}^{(2)},t^{(2)})]=2\delta(\|\bs{s}^{(1)}-\bs{s}^{(2)}\|,|t^{(1)}-t^{(2)}|).$$
Since we assume $\delta$ to depend only on the norm of $\bs s^{(1)}-\bs s^{(2)}$,  
the associated  process is (spatially) isotropic.

We assume the dependence function $\delta$ to be given for $v,u \geq 0$ by
\begin{align}\label{delta}
\delta(v,u) = 2\theta_1 v^{\alpha_1}+2\theta_2u^{\alpha_2},
\end{align}
where $0<\alpha_1,\alpha_2 \leq 2$ and $\theta_1,\theta_2>0$. 
This is the fractional class frequently used for dependence modelling, and here defined with respect to space and time.

The bivariate distribution function of $(\eta(\bs{0},0),\eta(\bs{h},u))$ is given for $x_1,x_2>0$ by
\begin{align}\label{bivhuesler}
F(x_1,x_2)=\exp\Bigg\{&-\frac{1}{x_1}\Phi\left(\frac{\log({x_2}/{x_1})}{\sqrt{2\delta(\|\bs{h}\|,|u|)}} + \sqrt{\frac{\delta(\|\bs{h}\|,|u|)}{2}}\right) \nonumber\\
&-\frac{1}{x_2}\Phi\left(\frac{\log({x_1}/{x_2})}{\sqrt{2\delta(\|\bs{h}\|,|u|)}}+ \sqrt{\frac{\delta(\|\bs{h}\|,|u|)}{2}}\right)\Bigg\},
\end{align}
where $\Phi$ denotes the standard normal distribution function (cf. \citet{Steinkohl}).

The parameters of interest are contained in the dependence function $\delta$. 
We refer to $(\theta_1,\alpha_1)$ as the {\em spatial parameter} and to $(\theta_2,\alpha_2)$ as the {\em temporal parameter}.
From the bivariate distribution function in \eqref{bivhuesler}, the pairwise density can be derived and pairwise likelihood methods can be used to estimate the parameters; cf. \citet{Steinkohl2}, \citet{Huser} and  \citet{Ribatet}. Full likelihood inference is virtually intractable in a general multidimensional setting, as the number of terms occurring in the likelihood explode. 
More recently, however, parametric inference methods based on higher-dimensional margins have been proposed that work in specific scenarios, see for instance \citet{Genton2}, who use triplewise instead of pairwise likelihood, \citet{Engelke1}, who propose a threshold-based approach, or \citet{Thibaud_Opitz} and \citet{Wadsworth_Tawn}, who use a censoring scheme for bias reduction.

In the following we apply the estimation method introduced in Section~\ref{Desmodel} based on the extremogram of more general regularly varying processes to the special case of the Brown-Resnick process~\eqref{limitfield}. We make use of the fact that its extremogram possesses a closed-form expression which is characterized by the dependence function $\delta$.

\begin{lemma}[\citet{Steinkohl}, equation (3.1)]\label{extspace}
Let $\{\eta(\bs s,t): (\bs s,t) \in \bbr^2 \times [0,\infty)\}$ be the strictly stationary {isotropic} Brown-Resnick process in $\R^2\times [0,\infty)$ as defined in \eqref{limitfield}  with dependence function 
given in \eqref{delta}.
Then the extremogram of $\eta$ is given by
\begin{equation}\label{chi_BR}
\chi(v,u) = 2\Big(1-\Phi\Big(\sqrt{\frac1{2}\delta(v,u)}\Big) \Big) = 2\big(1-\Phi(\sqrt{\theta_1 v^{\alpha_1}+\theta_2u^{\alpha_2}})\big), \quad v,u\geq 0.
\end{equation}
\end{lemma}

Solving equation \eqref{chi_BR} for $\delta(v,u)$ leads to
\begin{equation}
\frac{\delta(v,u)}{2} = \theta_1 v^{\alpha_1} + \theta_2 u^{\alpha_2} = \Big(\Phi^{-1}\big(1-\frac{1}{2}\chi(v,u)\big)\Big)^2.
\label{deltachi}
\end{equation}
For temporal lag $0$ and taking the logarithm on both sides we have
$$2\log\Big(\Phi^{-1}\big(1-\frac{1}{2}\chi(v,0)\big)\Big) = \log (\theta_1) + \alpha_1\log v=:\log (\theta_1) + \alpha_1x_v.$$
In the same way, we obtain
$$2\log\Big(\Phi^{-1}\big(1-\frac{1}{2}\chi(0,u)\big)\Big) =: \log (\theta_2) + \alpha_2 x_u.$$
To put this in the context of equations~\eqref{eq:param_space} and \eqref{eq:param_time}, first note that in the weighted linear regression, instead of working with the ``original'' lags $v$ and $u$, we consider their log transformations $x_v=\log(v)$ and $x_u=\log(u)$; hence in particular, we need to exclude the lags $v=0$ and $u=0$. The observation scheme described in Condition~\ref{grid} then yields that $u,v \geq 1$ and thus $x_v,x_u \geq 0$. We furthermore set $C_1=\log(\theta_1)$, $C_2=\log(\theta_2)$ and choose the transformations $T_1$ and $T_2$ defined by $T_1(\chi(v,0))=2\log\big(\Phi^{-1}\big(1-\frac{1}{2}\chi(v,0)\big)\big)$ and $T_2(\chi(0,u))=2\log\big(\Phi^{-1}\big(1-\frac{1}{2}\chi(0,u)\big)\big)$. The parameter spaces are given by $\Theta_{\cals}=\Theta_{\mathcal{T}}=\mathbb{R} \times (0,2]$.

In the following we work out necessary and sufficient conditions for the Brown-Resnick process~\eqref{limitfield} with dependence function~\eqref{delta} to satisfy the conditions of Theorem~\ref{asysemispace}, focusing again on the spatial case; i.e., on the processes $\left\{\eta(\bs{s},t): \bs{s}\in \bbr^2\right\}$ for fixed observed $t\in \{t_1,\ldots,t_T\}$. Furthermore we show how the fact that the model parameter $\alpha_1 \in (0,2]$ has bounded support influences the asymptotics of the WLSE $(\wh{\theta}_1,\wh{\alpha}_1)$.

\subsection{Asymptotics of the empirical spatial extremogram of the Brown-Resnick process}\label{spacesection1_BR}

For a start, we need a sufficiently precise estimate for the extremogram~\eqref{chi_BR} of the Brown-Resnick process, which we give now.

\ble\label{le3.2}
Let $\bs s,\bs h \in \mathbb{R}^2$.
For every sequence $a_n \to \infty$ we have for fixed $t \in [0,\infty),$
\begin{align*}
&\frac{\mathbb{P}(\eta(\bs{s},t)>a_n,\eta(\bs{s}+\bs{h},t)>a_n)}{\mathbb{P}(\eta(\bs{s},t)>a_n)} \\
 =&\chi(\|\bs h \|,0)+ \Big[\frac{1}{2a_n}\big(\chi(\|\bs h \|,0)-2\big)\big(\chi(\|\bs h \|,0)-1\big)\Big](1+o(1)).
\end{align*}
\ele

Lemma~\ref{le3.2} is a direct application of {Lemma~A.1}(b) of \citet{buhl3} for $A=B=(1,\infty)$ and their {equation~(A.4)}.
This applies since $\{\eta(\bs s,t): \bs s \in \bbr^2\}$ has finite-dimensional standard unit Fr{\'e}chet marginal distributions. 
We can choose in the following $a_n = n^2$ in order to satisfy the condition $\mathbb{P}(|\eta(\bs 0,0)| > a_n) \sim n^{-2}$ as $\nto$ from Definition~\ref{def:reg_var}. Recall furthermore that we have to choose a finite set $\calv=\{v_1,\ldots,v_p\}$ of observed lags, which show clear extremal dependence as explained in Section~\ref{Desmodel}.

\begin{theorem}\label{extspace2_BR} 
Consider the spatial Brown-Resnick process $\left\{\eta(\bs{s},t): \bs{s}\in \bbr^2\right\}$ as defined in \eqref{limitfield} {with dependence function given in \eqref{delta}}.
 Set $m_n=n^{\beta_1}$ for $\beta_1\in(0,1/2)$.
Then the empirical spatial extremogram $\wh{\chi}^{(t)}(v,0)$ defined in \eqref{extremogramspace} with the quantile $q=a_{m_n}=m_n^2$ satisfies
\begin{align}
\frac{n}{m_n}\big(\wh{\chi}^{(t)}(v,0) - \chi_n(v,0)\big)_{v \in \calv}\stackrel{d}{\to} \mathcal{N}(\bs 0,\Pi^{(\iso)}_1), \quad \nto, \label{CLT_BR}
\end{align}
where the covariance matrix $\Pi^{(\iso)}_1$ is specified in equation~\eqref{isomat}, and $\chi_n$ is the pre-asymptotic spatial extremogram as in \eqref{preasymptotic1}. 

Furthermore, for the averaged empirical extremogram $\wt{\chi}(v,0)=T^{-1}\sum_{k=1}^T\wt{\chi}^{{(t_k})}(v,0)$ defined in~\eqref{extrspacemean}, we obtain  (with covariance matrix $\Pi^{(\iso)}_2$ given in equation~\eqref{defPi})
\begin{align}
\frac{n}{m_n}\big(\wh{\chi}(v,0) - \chi_n(v,0)\big)_{v \in \calv}\stackrel{d}{\to} \mathcal{N}(\bs 0,\Pi^{(\iso)}_2), \quad \nto. \label{CLT_ave_BR}
\end{align}
\end{theorem}

\bproof
We need to verify the conditions of Corollary~\ref{coroll_average}; i.e., conditions (M1)-(M4) of Theorem~\ref{extspace2} for $a_{m_n}=m_n^2$, and apply results of Section~8 of the supplement~\cite{BDKSsupp}. \\[2mm]
Condition~{(M1)} is satisfied by equation~(8.2).\\[2mm] 
To show conditions~{(M2)-(M4)} we choose sequences $m_n=n^{\beta_1}$ and $r_n=n^{\beta_2}$ for $0<\beta_1<1/2$ and $0<\beta_2<\beta_1.$ 
For this choice $m_n$ and $r_n$ increase to infinity with $m_n={o(n)}$ and $r_n=o(m_n)$ as required.\\[2mm]
Condition~{(M2)}; i.e., $m_n^2r_n^2/n=n^{2(\beta_1+\beta_2)-1} \to 0$ holds if and only if $\beta_2 \in (0,\min\{\beta_1,(1/2-\beta_1)\})$.\\[2mm]
We now show condition~{(M3)}. Choose $\gamma > 0$, such that all lags in $\calv$  lie in $B(\bs 0,\ga)=\{\bs s\in\mathbb{Z}^2 : \|\bs s\|\le\ga\}$. 
{For $\eps>0$, like in Example~4.6 of \citet{buhl3}, we have for $\bs s, \bs s' \in \bbr^2$ by a Taylor expansion, 
\begin{align*}
\mathbb{P}(\eta(\bs s,t)>\eps m_n^2,\,&\eta(\bs s',t)>\eps m_n^2)\\&=1-2\mathbb{P}(\eta(\bs 0,0) \leq \eps m_n^2)+\mathbb{P}(\eta(\bs s,t)\leq \eps m_n^2,\eta(\bs s',t)\leq \eps m_n^2) \\
&=1-2\exp\Big\{-\frac{1}{x}\Big\}+\exp\Big\{-\frac{2-\chi(\|\bs s-\bs s'\|,0)}{\eps m_n^2}\Big\} \\
&=\frac{1}{\eps m_n^2}\chi(\|\bs s-\bs s'\|,0)+\mathcal{O}\Big(\frac{1}{m_n^4}\Big),\quad\nto.
\end{align*}
Therefore}, for $\|\bs h\| \geq 2\ga$, 
\begin{align}
\lefteqn{\mathbb{P}(\max\limits_{\bs s \in B(\bs 0, \ga)} \eta(\bs s,t) > \eps m_n^2,\max\limits_{\bs s' \in B(\bs h, \ga)} \eta(\bs s',t) > \eps m_n^2)} \nonumber\\
\leq & \sum\limits_{\bs s \in B(\bs 0, \ga)}\sum\limits_{\bs s' \in B(\bs h, \ga)}\mathbb{P}(\eta(\bs s,t) > \eps m_n^2, \eta(\bs s',t) > \eps m_n^2) \nonumber\\
= & \sum\limits_{\bs s \in B(\bs 0, \ga)}\sum\limits_{\bs s' \in B(\bs h, \ga)}\Big\{\frac1{\eps m_n^2}\chi(\|\bs s-\bs s'\|,0) +\calo\Big(\frac{1}{m_n^4}\Big)\Big\}\nonumber\\
\leq &\frac{2|B(\bs 0, \ga)|^2}{\eps m_n^2} \big(1-\Phi(\sqrt{\theta_1(\|\bs h\|-2\ga)^{\alpha_1}}\big) + \mathcal{O}\Big(\frac{1}{m_n^4}\Big), \label{spat_cond3}
\end{align}
as $\nto$, where we have used \eqref{chi_BR}. 
Summarize $V:=\{v=\|\bs h\|: \bs h \in \mathbb{Z}^2\}$ and note that $|\{\bs h \in \mathbb{Z}^2: \|\bs h\|=v\}|=\mathcal{O}(v)$.
Therefore, for $k\ge 2\ga$,
\begin{align*}
L_{m_n} :=& \, \limsup_{\nto} m_n^2 \sum_{\bs h \in \mathbb{Z}^{2} \atop k< \|\bs h\| \leq r_n} 
\mathbb{P}\Big(\max\limits_{\bs s \in B(\bs 0, \ga)} \eta(\bs s,t)>\eps m_n^2, \max\limits_{\bs s' \in B(\bs h, \ga)} \eta(\bs s',t)>\eps m_n^2\Big) \\
\leq & \, 2 |B(\bs 0, \ga)|^2\limsup_{\nto} \bigg\{\sum_{\bs h \in \mathbb{Z}^{2} \atop k< \|\bs h\| \leq r_n} 
\Big\{ \frac{1}{\eps} (1-\Phi(\sqrt{\theta_1(\|\bs h\|-2\ga)^{\alpha_1}}))\Big\}+\mathcal{O}\Big(\Big(\frac{r_n}{m_n}\Big)^2\Big)\bigg\}\\
\leq &\,  K_1 
\limsup\limits_{\nto} \sum\limits_{v\in V: \atop k< v \leq r_n} \Big\{  \frac{v}{\eps} 2\big(1-\Phi(\sqrt{\theta_1(v-2\ga)^{\alpha_1}})\big)\Big\},
\end{align*}
for some constant $K_1>0$.  
For the term $\mathcal{O}((r_n/m_n)^2)$ we use that $r_n/m_n \rightarrow 0$.
From Lemma~8.3 
and the fact that $1-\Phi(x) \leq \exp\{-x^2/2\}$ for $x>0$, we find for $K_2>0$, 
\begin{align*}
L_{m_n} \le &  {K_2 } k^2 \exp\big\{-\frac12 \theta_1(k-2\ga)^{\alpha_1}\big\}.
\end{align*}
Since $\alpha_1>0$, the right hand side converges to 0 as $k \rightarrow \infty$ ensuring condition~{(M3)}.\\[2mm]
Now we turn to the mixing conditions~{(M4)}.\\
We start with~{(M4i)}.
 With $V$ as before, and with equation~(8.2), 
 we estimate, recalling from above that the number of lags $\|\bs h\|=v$ is of oder $\mathcal{O}(v)$,
\begin{align*}
m_n^2\sum\limits_{\bs h \in \mathbb{Z}^{2} : \Vert{\bs h}\Vert > r_n} \alpha_{1, 1}(\|\bs h\|)
&\leq K_1 m_n^2\sum\limits_{v \in V : v > r_n} v\,\alpha_{1,1}(v)
 \leq 4 K_1  m_n^2 \sum\limits_{v \in V : v > r_n}  v\,e^{-{\theta_1 v^{\alpha_1}}/2}.
\end{align*}
By Lemma~8.3 
we find 
$$m_n^2\sum\limits_{v \in V : v > r_n} v \, e^{-{\theta_1 v^{\alpha_1}}/2} 
\leq c m_n^2 r_n^2 \ e^{-{\theta_1r_n^{\alpha_1}}/2}
=c m_n^2 r_n^2 \ e^{- {\theta_1n^{\alpha_1\beta_2}}/2}\to 0,\quad \nto.$$
By the same arguments condition~{(M4ii)} is satisfied. \\
Condition~{(M4iii)} holds by equation~(8.2), 
since
\begin{align*}
{m_n \, n\,  \alpha_{1,n^2}(r_n)} \leq 4 n^3 m_n \, e^{- {\theta_1r_n^{\alpha_1}}/2}\to 0,\quad \nto. 
\end{align*}
\eproof

\brem\label{remarkextspace_Fr}
We want to examine for which choices of $\beta_1$, introduced with the sequence $m_n=n^{\beta_1}$ in Theorem~\ref{extspace2_BR}, we can replace the pre-asymptotic extremogram by the theoretical one in the CLTs~\eqref{CLT_BR} and \eqref{CLT_ave_BR}; that is, the bias condition~\eqref{condition5},
\begin{equation*}
\frac{n}{m_n}\left(\chi_n(v,0) - \chi(v,0)\right) \to 0,\quad \nto,
\end{equation*}
is satisfied for all spatial {lags} $v\in \calv$.
For the Brown-Resnick process \eqref{limitfield} we obtain from Lemma~\ref{le3.2}, 
\begin{align*}
&\frac{n}{m_n}\left(\chi_n(v,0)- \chi(v,0)\right) \\
&{} = \frac{n}{m_n}\left(\frac{\mathbb{P}(\eta(\bs{s},t)>m_n^{2},\eta(\bs{s}+\bs{h},t)>m_n^{2})}{\mathbb{P}(\eta(\bs{s},t)>m_n^{2})} - \chi(v,0)\right) \\
&{} \sim  \frac{n}{2m_n^3}\big(\chi(v,0)-2\big)\big(\chi(v,0)-1\big)\\
&{} = n^{1-3\beta_1} \frac{1}{2}\big(\chi(v,0)-2\big)\big(\chi(v,0)-1\big)\quad
 \to \quad 0 \quad\text{  if and only if } \beta_1 >1/3;
 \end{align*}
{cf. Theorem~4.4 of \citet{buhl3}.}
 Thus we have to distinguish two cases:
 \begin{enumerate}
 \item[(I)]
For  $\beta_1 \leq 1/3$ we cannot replace the pre-asymptotic extremogram by the theoretical version, but can resort to a bias correction, which is decribed in \eqref{biascorrection} below.
\item[(II)]
 For $1/3<\beta_1<1/2$ we obtain indeed
 \beam\label{ohnebias}
 n^{1-\beta_1}\big(\wh{\chi}^{(t)}(v,0) - \chi(v,0)\big)_{v \in \calv}\stackrel{d}{\to} \mathcal{N}(\bs 0,\Pi^{(\textnormal{iso})}_1), \quad \nto,
 \eeam 
and likewise for the averaged empirical extremogram,
 \beam\label{ohnebias_2}
 n^{1-\beta_1}\big(\wh{\chi}(v,0) - \chi(v,0)\big)_{v \in \calv}\stackrel{d}{\to} \mathcal{N}(\bs 0,\Pi^{(\textnormal{iso})}_2), \quad \nto.
 \eeam
 \end{enumerate}
\erem

We now turn to the bias correction needed in case (I).  
By Lemma~\ref{le3.2} the pre-asymptotic extremogram has representation 
\begin{align}
\chi_n(v,0)&=\chi(v,0)+ \Big[\frac{1}{2m_n^{2}}\big(\chi(v,0)-2\big)\big(\chi(v,0)-1\big)\Big](1+o(1)) \nonumber\\
&=\chi(v,0) + \frac{1}{2m_n^{2}}\nu(v,0)(1+o(1)),\quad \nto, \label{chim}
\end{align}
where $\nu(v,0):=\big(\chi(v,0)-2\big)\big(\chi(v,0)-1\big).$
Consequently, we propose for fixed $t \in \{t_1,\ldots,t_T\}$ {and all $v\in\calv$}  the {\em bias corrected empirical  spatial extremogram } 
\begin{align*}
\wh{\chi}^{(t)}(v,0)-\frac{1}{2m_n^{2}}\big(\wh{\chi}^{(t)}(v,0)-2\big)\big(\wh{\chi}^{(t)}(v,0)-1\big)
=:\wh{\chi}^{(t)}(v,0)-\frac{1}{2m_n^{2}}\wh{\nu}^{(t)}(v,0), 
\end{align*}
and set 
\begin{align}
\wt{\chi}^{(t)}(v,0) := 
\begin{cases}
\wh{\chi}^{(t)}(v,0)-\dfrac{1}{2m_n^{2}}\wh{\nu}^{(t)}(v,0)\quad & \mbox{if }  m_n=n^{\beta_1} \text{ with } \beta_1 \in (\frac1{5},\frac1{3}], \\
\wh{\chi}^{(t)}(v,0)\quad & \mbox{if } m_n=n^{\beta_1} \text{ with } \beta_1 \in (\frac1{3},\frac1{2}). \label{biascorrection}
\end{cases} 
\end{align}
Theorem~\ref{biastheorem} below shows asymptotic normality of the bias corrected extremogram centred by the true one and, in particular, why $\beta_1$ has to be larger than $1/5$.

\begin{theorem}\label{biastheorem}
For a fixed time point $t\in\{t_1,\ldots,t_T\}$ consider the spatial Brown-Resnick process $\left\{\eta(\bs{s},t),\bs{s}\in \bbr^2\right\}$ defined in \eqref{limitfield} {with dependence function given in \eqref{delta}}. 
Set $m_n=n^{\beta_1}$ for $\beta_1\in\big(\frac{1}{5},\frac{1}{3}\big]$.
Then the bias corrected empirical spatial extremogram \eqref{biascorrection} satisfies
\beam\label{mitbias}
\frac{n}{m_n}\big(\wt{\chi}^{(t)}(v,0) - \chi(v,0)\big)_{v \in \calv}\stackrel{d}{\to} \mathcal{N}(\bs 0,\Pi^{(\textnormal{iso})}_1), \quad \nto,
\eeam
where $\Pi^{(\textnormal{iso})}_1$ is the covariance matrix as given in equation~\eqref{isomat}.
Furthermore, the corresponding bias corrected averaged version $\wt{\chi}(v,0)=T^{-1}\sum_{k=1}^T\wt{\chi}^{{(t_k})}(v,0)$ satisfies
$$\frac{n}{m_n} \Big(\wt{\chi}(v,0) - \chi(v,0)\Big)_{v \in \calv}\stackrel{d}{\to} \mathcal{N}(\bs 0,\Pi^{(\textnormal{iso})}_2), \quad \nto,$$
with covariance matrix $\Pi^{(\textnormal{iso})}_2$ specified in \eqref{defPi}.
\end{theorem}

\begin{proof}
For simplicity we suppress the time point $t$ in the notation.
By \eqref{chim} and \eqref{biascorrection} we have as $\nto$, $$\frac{n}{m_n}(\wt{\chi}(v,0) - \chi(v,0)) \sim \frac{n}{m_n}(\wh{\chi}(v,0)-\chi_n(v,0)) - \frac{n}{2m_n^3}(\wh{\nu}(v,0)-\nu(v,0)).$$
By Theorem \ref{extspace2_BR} it suffices to show that $(n/(2m_n^3))(\wh{\nu}(v,0)-\nu(v,0)) \stackrel{P}{\to} 0$. 
Setting $\nu_n(v,0) := \big(\chi_n(v,0)-2\big)\big(\chi_n(v,0)-1\big)$ we have
\begin{align*}
\frac{n}{2m_n^3}(\wh{\nu}(v,0)-\nu(v,0))=\frac{n}{2m_n^3} (\wh{\nu}(v,0)-\nu_{n}(v,0)) + \frac{n}{2m_n^3}(\nu_{n}(v,0)-\nu(v,0))\eqqcolon A_1+A_2.
\end{align*}
We calculate 
\begin{align*}
&\frac{n}{m_n (2\chi(v,0)-3)}\Big(\wh{\nu}(v,0)-\nu_{n}(v,0)\Big)\\
&=\frac{n}{m_n(2\chi(v,0)-3)}\Big(\wh{\chi}^2(v,0)-3\wh{\chi}(v,0)-(\chi_n^2(v,0)-3\chi_n(v,0))\Big)\\
&= \frac{n}{m_n(2\chi(v,0)-3)}
\Big((\wh{\chi}(v,0) -\chi_n(v,0))(\wh{\chi}(v,0) +\chi_n(v,0)) -3(\wh{\chi}(v,0)-\chi_n(v,0))\Big)\\
&=\frac{n}{m_n}
\Big(\wh{\chi}(v,0) -\chi_n(v,0)\Big) \frac{\wh{\chi}(v,0) +\chi_n(v,0) -3}{2\chi(v,0)-3}.
\end{align*}
The first term
converges by Theorem \ref{extspace2_BR} weakly to a normal distribution, and
the second term, 
together with the fact that $\wh{\chi}(v,0)\stackrel{P}{\to} \chi(v,0)$ and $\chi_n(v,0)\stackrel{P}{\to} \chi(v,0)$, converges to 1 in probability.
Hence, it follows from Slutzky's theorem that $A_1\stackrel{P}{\to} 0$.
Now we turn to $A_2$ and calculate
\begin{align*}
&\nu_{n}(v,0) =\chi_n^2(v,0)-3\chi_n(v,0)+2 \\
 &\quad \sim \Big(\chi(v,0)+\frac{1}{2m_n^{2}}\nu(v,0)\Big)^2 -3\Big( \chi(v,0)+\frac{1}{2m_n^{2}}\nu(v,0)\Big)+2 \\
&\quad= \chi^2(v,0)- 3\chi(v,0) +2+\frac{1}{m_n^{2}}\chi(v,0)\nu(v,0)+ \frac{1}{4m_n^{4}}\nu(v,0)^2 -\frac{3}{2m_n^{2}}\nu(v,0)\\
&\quad= \big(\chi(v,0)-2\big)\big(\chi(v,0)-1\big)+\frac{1}{m_n^{2}}\chi(v,0)\nu(v,0)+ \frac{1}{4m_n^{4}}\nu(v,0)^2 -\frac{3}{2m_n^{2}}\nu(v,0)\\
&\quad= \nu(v,0) + \frac{\nu(v,0)}{m_n^{2}}\Big(\chi(v,0) +\frac{1}{4m_n^{2}}\nu(v,0)-\frac{3}{2}\Big),
\end{align*}
where we have used \eqref{chim}.
Therefore, $A_2$ converges to $0$, if ${n}/{m_n^5}\to 0$ as $\nto$.
With $m_n=n^{\beta_1}$ it follows that $\beta_1>\frac{1}{5}$.  
Finally, the last statement follows as Corollary~\ref{coroll_average}.
\end{proof}

\brem\label{rates}
{Note that in~\eqref{ohnebias} and \eqref{ohnebias_2} the rate of convergence is of the order $n^a$ for $a\in(1/2, 2/3)$. } 
{On the other hand, after bias correction in \eqref{mitbias} we obtain convergence of the order $n^a$ for $a\in [2/3,4/5)$; i.e. a better rate. } 
\erem

\bexam
We generate 100 realizations of the Brown-Resnick process in \eqref{limitfield} using the \texttt{R}-package \texttt{RandomFields} \cite{Schlather5} and the exact method {via extremal functions} proposed in \citet{Dombry2}, {Section~2}. 
We then compare the empirical estimates  of the spatial extremogram $\wh{\chi}(v,0)$ in \eqref{extremogramspace} and the bias corrected ones $\wt{\chi}(v,0)$ in \eqref{biascorrection} with the true theoretical extremogram $\chi(v,0)$ {for lags $v \in \{1,\sqrt{2},2,\sqrt{5},\sqrt{8},3,\sqrt{10},\sqrt{13},4,\sqrt{17}\}$}. 
We choose the parameters $\theta_1 =0.4$ and $\alpha_1=1.5$. 
The grid size and the number of time points are given by $n=70$ and $T=10$. 
The results are summarized in Figure \ref{Bias}. 
We see that the bias corrected extremogram is closer to the true one.
\eexam

\begin{figure}[t!] 
\centering
\subfloat[]{\includegraphics[scale=0.22]{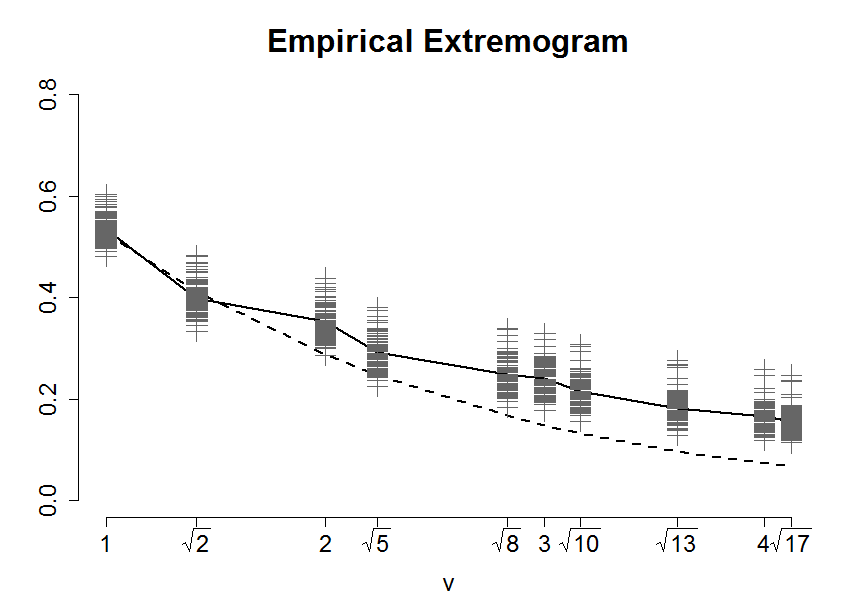}}
\subfloat[]{\includegraphics[scale=0.22]{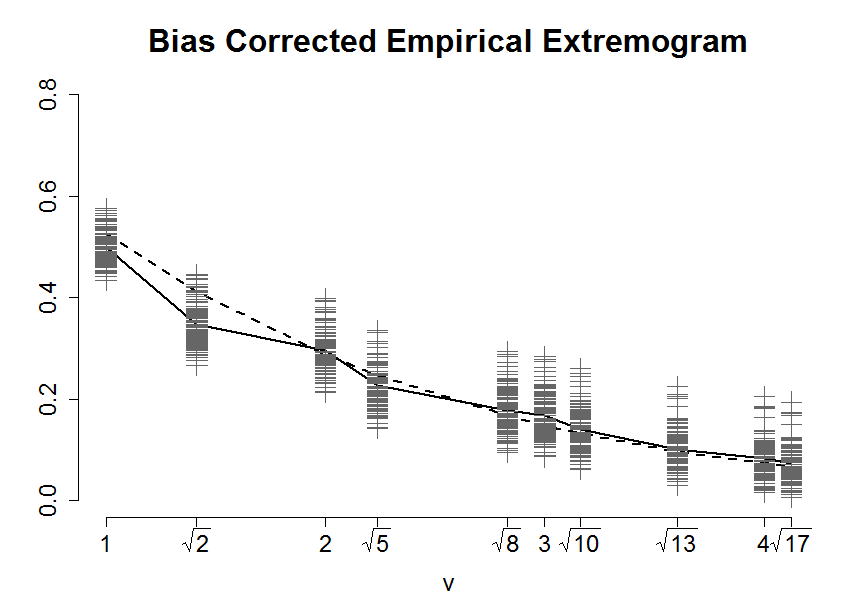}}
\caption{Empirical spatial extremogram (left) and its bias corrected version  (right) for 100 simulated max-stable random fields in \eqref{limitfield} with $\delta(v,0) = 2\cdot 0.4 v^{1.5}$. The dashed line represents the theoretical spatial extremogram and the solid line is the mean over all 100 replicates.}
\label{Bias}
\end{figure}

\subsection{Asymptotic properties of spatial parameter estimates of the Brown-Resnick process}\label{spacesection2_BR}
In this section we prove asymptotic normality of the {WLSE $(\wh\theta_1,\wh\alpha_1)$}. We proceed as in the more general setting in Section~\ref{spacesection2}. Recall that in the more specific situation here we have $C_1=\log(\theta_1)$ and choose the transformation $T_1(\chi(v,0))=2\log\big(\Phi^{-1}\big(1-\frac{1}{2}\chi(v,0)\big)\big)$, where the log transformed version of the spatial lag satisfies $x_v=\log(v) \geq 0$ for $v \in \calv$. We set $\wt{\chi}(v,0) = \frac{1}{T}\sum_{k=1}^T \wt{\chi}^{(t_k)}(v,0)$ as in~\eqref{extrspacemean}, possibly after a bias correction, which depends on the two cases described in Remark~\ref{remarkextspace_Fr}. 
The analogue of the weighted least squares optimization problem \eqref{minspace}  then reads as 
\begin{align}\label{extrspaceneu_BR}
\begin{pmatrix}\wh{\theta}_1 \\ 
\wh{\alpha}_1
\end{pmatrix} 
= \argmin_{\stackrel{\theta_1,\alpha_1>0}{\alpha_1\in (0,2] }} 
&\sum_{v \in \calv}  w_v  \big(T_1(\wt{\chi}(v,0))-\big(\log(\theta_1) + \alpha_1 x_v\big)\big)^2.
\end{align}
Note in particular that $\alpha_1$ has bounded support; this has to be treated as a special case in what follows.
To show asymptotic normality of the WLSE in \eqref{extrspaceneu_BR}, we define as before  design and weight matrices $X$ and $W$ as
$$X = [\bs{1},(x_v)\trans_{v \in \calv}]\in\R^{p\times 2} \quad \text{ and } \quad W=\diag\{w_v: v \in \calv\}\in\R^{p\times p},$$
respectively, where $\bs{1}=(1,\ldots,1)\trans \in \mathbb{R}^{p}$.
Let $\bs{\psi}_1 = (\log(\theta_1),\alpha_1)\trans$ be the parameter vector with parameter space $\Theta_{\cals} = \bbr\times(0,2]$. 
Then the WLSE; i.e., the solution to \eqref{extrspaceneu_BR} is given by 
\beam\label{psihat}
\wh{\bs{\psi}}_1=\begin{pmatrix}\log(\wh{\theta}_1) \\ \wh{\alpha}_1
\end{pmatrix}= (X\trans W X)^{-1}X\trans W (T_1(\wt{\chi}(v,0)))\trans_{v \in \calv}.
\eeam
Without any constraints $\wh{\bs{\psi}}_1$ may produce estimates of $\alpha_1$ outside its parameter space $(0,2]$. 
In such cases  we set the parameter estimate equal to 2, and we denote the resulting estimate by $\bs{\wh{\psi}}_1^c = (\log(\wh \theta_1^c),\wh{\alpha}_1^c)\trans$.

\begin{theorem}\label{asysemispace_BR}
Let $\bs{\wh{\psi}}_1^c = (\log(\wh \theta_1^c),\wh{\alpha}_1^c)\trans$ denote the WLSE resulting from the constrained minimization problem \eqref{extrspaceneu_BR} and $\bs{\psi}^{*}_1=(\log(\theta_1^*),\alpha_1^*)\trans\in \Theta_{\cals}$ the true parameter vector. 
Set $m_n = n^{\beta_1}$ for $\beta_1\in(1/5,1/2)$.
Then as $\nto$,
\begin{equation}
\frac{n}{m_n}\Big(\bs{\wh{\psi}}_1^c - \bs{\psi}_1^{*}\Big) \stackrel{d}{\to} 
\begin{cases}
\bs Z_1\quad & \mbox{if } \alpha_1^{*}<2, \\
\bs Z_2\quad & \mbox{if }  \alpha_1^{*}=2,
\end{cases} 
\end{equation}
where $\bs Z_1 \sim \mathcal{N}(\bs 0,\Pi^{(\textnormal{iso})}_3)$, and the distribution of $\bs Z_2$ is given by
\begin{align}
& \mathbb{P}\left(\bs Z_2 \in B\right)
= \int_{B\cap \left\{(b_1,b_2)\in \bbr^2:b_2<0\right\}} \varphi_{\bs{0},\Pi^{(\textnormal{iso})}_3}(z_1,z_2)dz_1dz_2\label{DistZ2}\\
&\qquad \qquad+ \int_{0}^{\infty}\int_{\{b_1\in\bbr: (b_1,0)\in B\}} \varphi_{\bs{0},\Pi^{(\textnormal{iso})}_3}\Big (z_1-\frac1{\sum_{v \in \calv} w_v}\sum_{v \in \calv} (w_vx_v) \  z_2,\,z_2\Big)dz_1dz_2\nonumber
\end{align}
for every Borel set $B$ in $\R^2$, and $\varphi_{\bs 0,\Sigma}$ denotes the bivariate normal density with mean vector $\bs 0$ and covariance matrix $\Sigma$. In particular, the joint  distribution function of $\bs Z_2$ is given for $(p_1,p_2)\trans\in \mathbb{R}^2$ by
\begin{align}
&\mathbb{P}\left(\bs Z_2 \leq (p_1,p_2)\trans\right)= \int\limits_{-\infty}^{\min\{0,p_2\}} \int\limits_{-\infty}^{p_1} \varphi_{\bs{0},\Pi^{(\textnormal{iso})}_3}(z_1,z_2)dz_1dz_2 \label{cdfZ2}\\
&\qquad \qquad+ \mathbbmss{1}_{\{p_2 \geq 0\}} \int\limits_{0}^{\infty}\int\limits_{-\infty}^{p_1} \varphi_{\bs{0},\Pi^{(\textnormal{iso})}_3}\Big (z_1-\frac1{\sum_{v \in \calv} w_v}\sum_{v \in \calv} (w_vx_v) \  z_2,\,z_2\Big)dz_1dz_2. \nonumber
\end{align}
The covariance matrix of $\bs Z_1$ has representation
\begin{equation}
\Pi^{(\textnormal{iso})}_3 = Q_x^{(w)}G\Pi^{(\textnormal{iso})}_2 G{Q_x^{(w)}}\trans,
\label{covspace3}
\end{equation}
where $\Pi^{(\textnormal{iso})}_2$ is the covariance matrix given in \eqref{defPi}, 
\begin{align*}
Q_x^{(w)} &= (X\trans W X)^{-1}X\trans W \quad \text{ and}\quad
G = \diag\bigg \{ \sqrt{\frac{2\pi}{\theta_1^*v^{\alpha_1^*}}}  \exp\Big\{\frac1{2}\theta_1^*v^{\alpha_1^*}\Big\}: \ v \in \calv\hspace*{0.1cm}\bigg \}.
\end{align*}
\end{theorem}

\begin{proof}
For the first part of the proof, we neglect the constraints on $\alpha_1$. Then we can directly use Theorem~\ref{asysemispace}, observing that the derivative of $T_1$ is given by
$$T_1'(x) = -\Big(\Phi^{-1}(1-\frac{x}{2})\,\varphi(\Phi^{-1}(1-\frac{x}{2}))\Big)^{-1}, \quad 0<x<1,$$
where $\varphi$ is the univariate standard normal density. Thus,
$$T_1'(\chi(v,0)) = -\Big(\sqrt{\theta_1^*v^{\alpha_1^*}}\,\varphi\big(\sqrt{\theta_1^*v^{\alpha_1^*}}\big)\Big)^{-1}=
 -\sqrt{\frac{2\pi}{\theta_1^*v^{\alpha_1^*}}}  \exp\Big\{\frac1{2}\theta_1^*v^{\alpha_1^*}\Big\}.$$
Hence, as $\nto$,
$$\frac{n}{m_n}\left(\bs{\wh{\psi}}_1-\bs{\psi}_1^*\right)= \frac{n}{m_n} Q_x^{(w)} \big(T_1(\wt{\chi}(v,0))-T_1(\chi(v,0))\big)_{v \in \calv} \stackrel{d}{\to} \mathcal{N}\left(\bs{0},Q_x^{(w)}G\Pi^{(\textnormal{iso})}_2 G {Q_x^{(w)}}\trans\right).$$ Note that we can define the diagonal matrix $G$ unsigned, since signs cancel out.
We now turn to the constraints on $\alpha_1$.
Since the objective function is quadratic, if the unconstrained estimate exceeds two, the constraint $\alpha_1\in(0,2]$ results in an estimate $\wh{\alpha}_1^c=2$.
We consider separately the cases $\alpha_1^{*}<2$ and $\alpha_1^{*}=2$; i.e., the true parameter lies either in the interior or on the boundary of the parameter space.
The constrained estimator $\bs{\wh{\psi}}_1^c$ can be written as
$$\bs{\wh{\psi}}_1^c = \bs{\wh{\psi}}_1\mathds{1}_{\left\{\wh{\alpha}_1\leq 2\right\}} + (\wh{\theta}_1,2)\trans \mathds{1}_{\left\{\wh{\alpha}_1> 2\right\}}. $$
We calculate the asymptotic probabilities for the events $\{\wh{\alpha}_1 \leq 2\}$ and $\{\wh{\alpha}_1 >2\}$,
\begin{align*}
\mathbb{P}(\wh{\alpha}_1\leq 2) &= \mathbb{P}\Big(\frac{n}{m_n}(\wh{\alpha}_1-\alpha_1^{*}) 
\leq \frac{n}{m_n}(2-\alpha_1^{*})\Big).
\end{align*}
Since for $\alpha_1^{*}<2$ as $\nto$
$$\frac{n}{m_n}\big(\wh{\alpha}_1-\alpha_1^{*}\big) \stackrel{d}{\to} 
\mathcal{N}\left(0,{(0,1)}
\Pi^{(\textnormal{iso})}_3 {(0,1)\trans} \right)\quad\mbox{and}\quad
\frac{n}{m_n}(2-\alpha_1^{*})\to \infty,
$$
it follows that
\begin{equation}\label{alpha1not0}
\mathbb{P}(\wh{\alpha}_1\leq 2) \to 1 \quad \text{and} \quad \mathbb{P}(\wh{\alpha}_1>2) \to 0, \quad \nto.
\end{equation}
Therefore, for $\alpha_{1}^{*}<2$,
$$\frac{n}{m_n}\big(\bs{\wh{\psi}}^c_1 - \bs{\psi}_1^{*}\big) \stackrel{d}{\to} \mathcal{N}(\bs 0,\Pi^{\textnormal{(iso)}}_3), \quad \nto.$$
We now consider the case $\alpha_1^{*} = 2$ and $\wh{\alpha}_1>2$ (the unconstrained estimate exceeds 2).
In this case {\eqref{extrspaceneu_BR} leads to} the constrained optimization problem
\begin{align*}
&\min_{\bs{\psi}_1}\{[W^{1/2}((T_1(\wt{\chi}(v,0)))\trans_{v \in \calv}-X\bs{\psi}_1)]\trans[W^{1/2}((T_1(\wt{\chi}(v,0)))\trans_{v \in \calv}-X\bs{\psi}_1]\}, \\ &\quad \text{s.t.} \quad (0,1)\bs{\psi}_1 =2.
\end{align*}
To obtain asymptotic results for $\bs{\wh{\psi}}^c_1 - \bs{\psi}_1^{*}$, the vector $\bs{\wh{\psi}}_1-\bs{\psi}_1^{*}$ is projected onto the line $\Lambda =\{\bs{\psi}\in \bbr^2, (0,1)\bs{\psi}=0\}$, i.e., denoting by $I_2$ the $2\times 2$-identity matrix, the projection matrix  with respect to the induced norm $\bs \psi \mapsto (\bs \psi\trans X\trans W X \bs \psi)^{1/2}$ is given by (cf. \citet{Andrews}, page 1365)
$$P_{\Lambda} = I_2-(X\trans W X)^{-1}(0,1)\trans((0,1)(X\trans W X)^{-1}(0,1)\trans)^{-1}(0,1).$$ 
For simplicity we use the abbreviation $p_{wx} = \sum_{v \in \calv}w_vx_v/\sum_{v \in \calv}w_v$.
We calculate
\begin{align*}
&(\bs{\wh{\psi}}^c_1-\bs{\psi}^{*}_1)\mathds{1}_{\{\wh{\alpha}_1>2\}} = P_{\Lambda}(\bs{\wh{\psi}}_1-\bs{\psi}_1^{*})\mathds{1}_{\{\wh{\alpha}_1>2\}} \\
 &= (\bs{\wh{\psi}}_1-\bs{\psi}_1^{*})\mathds{1}_{\{\wh{\alpha}_1>2\}} - (X\trans WX)^{-1}(0,1)\trans \left((0,1)(X\trans WX)^{-1}(0,1)\trans\right)^{-1}(\wh{\alpha}_1-2)\mathds{1}_{\{\wh{\alpha}_1>2\}}\\
&= (\bs{\wh{\psi}}_1-\bs{\psi}_1^{*})\mathds{1}_{\{\wh{\alpha}_1>2\}} + \begin{pmatrix}p_{wx} \\-1 \end{pmatrix}(\wh{\alpha}_1-2)\mathds{1}_{\{\wh{\alpha}_1>2\}}.
\end{align*}
For the joint constrained estimator $\bs{\psi}_1^c$ we obtain
\begin{align*}
\bs{\wh{\psi}}^c_1-\bs{\psi}_1^{*} &= (\bs{\wh{\psi}}^c_1-\bs{\psi}_1^{*})\mathds{1}_{\{\wh{\alpha}_1\leq 2\}} + (\bs{\wh{\psi}}^c_1-\bs{\psi}_1^{*})\mathds{1}_{\{\wh{\alpha}_1 >2\}} \\
&= (\bs{\wh{\psi}}_1-\bs{\psi}_1^{*})\mathds{1}_{\{\wh{\alpha}_1\leq 2\}} + (\bs{\wh{\psi}}_1-\bs{\psi}_1^{*})\mathds{1}_{\{\wh{\alpha}_1>2\}} + \begin{pmatrix}p_{wx}\\-1 \end{pmatrix}(\wh{\alpha}_1-2)\mathds{1}_{\{\wh{\alpha}_1>2\}} \\
& = (\bs{\wh{\psi}}_1-\bs{\psi}_1^{*})+ \begin{pmatrix}p_{wx}\\-1 \end{pmatrix}(\wh{\alpha}_1-2)\mathds{1}_{\{\wh{\alpha}_1>2\}}.
\end{align*}
This implies
\begin{align*}
\frac{n}{m_n}(\bs{\wh{\psi}}^c_1-\bs{\psi}_1^{*}) &=
\frac{n}{m_n}\begin{pmatrix}(\log(\wh \theta_1)-\log(\theta_1^{*})) + p_{wx} (\wh{\alpha}_1-2)\mathds{1}_{\{\wh{\alpha}_1>2\}} \\ (\wh{\alpha}_1-2) - (\wh{\alpha}_1-2)\mathds{1}_{\{\wh{\alpha}_1>2\}}\end{pmatrix}.
\end{align*}
Let $f(x_1,x_2) = (x_1+p_{wx}x_2 \mathds{1}_{\{x_2>0\}},x_2-x_2\mathds{1}_{\{x_2>0\}})\trans$ and observe that $f(c(x_1,x_2))=cf(x_1,x_2)$ for $c \geq 0.$
For the asymptotic distribution we calculate, denoting by $f^{-1}$ the inverse image of $f$,
\begin{align*}
&\mathbb{P}\Big(\frac{n}{m_n}(\bs{\wh{\psi}}^c_1-\bs{\psi}_1^{*}) \in B\Big) \\
&= \mathbb{P}\Big(\frac{n}{m_n}f(\bs{\wh{\psi}}_1-\bs{\psi}_1^{*})\in B\Big) = \mathbb{P}\Big(f\big(\frac{n}{m_n}(\bs{\wh{\psi}}_1-\bs{\psi}_1^{*})\big)\in B\Big)\\
& = \mathbb{P}\Big(\frac{n}{m_n}(\bs{\wh{\psi}}_1-\bs{\psi}_1^{*}) \in f^{-1}(B\cap \{(b_1,b_2)\in\bbr^2: b_2<0\}) \cup f^{-1}(B\cap \{(b_1,0): b_1\in\bbr\}) \Big) \\
&=\mathbb{P}\Big(\frac{n}{m_n}(\bs{\wh{\psi}}_1-\bs{\psi}_1^{*}) \in [B\cap \{(b_1,b_2)\in\bbr^2: b_2<0\}]\\
&\quad\quad  \cup  [\{(b_1-p_{wx}b_2,b_2), b_2\geq 0, (b_1,0) \in B\}]\Big) \\
&\to \int_{B\cap \left\{(b_1,b_2)\in \bbr^2,b_2<0\right\}} \varphi_{\bs{0},\Pi^{(\textnormal{iso})}_3}(z_1,z_2)dz_1dz_2 \\
&\qquad \qquad+ \int_{0}^{\infty}\int_{\{b_1\in\bbr, (b_1,0)\in B\}} \varphi_{\bs{0},\Pi^{(\textnormal{iso})}_3}(z_1-p_{wx}z_2,z_2)dz_1dz_2, \quad \nto.
\end{align*}
Plugging in $B=(-\infty,p_1] \times (-\infty,p_2]$ and using the Fubini-Tonelli theorem yields \eqref{cdfZ2}.
\end{proof}

\brem
The asymptotic properties for the constrained estimate are derived as a special case of Corollary 1 in Andrews \cite{Andrews}, who shows asymptotic properties of parameter estimates in a very general setting, when the true parameter is on the boundary of the parameter space.
The asymptotic distribution of the estimates for $\alpha_1^{*} =2 $ results from the fact that approximately half of the estimates lie above the true value and are therefore equal to two.
\erem

\section{Analysis of radar rainfall measurements}\label{sec:Florida}

Finally, we apply the Brown-Resnick space-time process in \eqref{limitfield} and the WLSE to radar rainfall data provided by the Southwest Florida Water Management District (SWFWMD)\footnote{http://www.swfwmd.state.fl.us/}. 
Our objective is to quantify their extremal behaviour 
by using spatial and temporal block maxima {and fitting a Brown-Resnick space-time process to the block maxima.}

The data base consists of radar values in inches measured on a $120\times120$km region containing 3600 grid locations. 
We calculate the spatial and temporal maxima over subregions of size $10\times 10$km and over 24 subsequent measurements of the corresponding hourly accumulated time series in the wet season (June to September) from the years 1999-2004. 
In this way we obtain $12\times 12$ locations on $732$ days of space-time block maxima of rainfall observations.
Taking block maxima  yields a process  consistent with the assumption of a max-stable process, or at least to lie in the domain of attraction of a max-stable process.  Taking daily data, we can furthermore ignore diurnal patterns.

We denote the set of locations by
$\cals=\{(i_1,i_2), i_1,i_2 \in \{1,\ldots,12\}\}$ and the space-time observations by
$\{\eta(\bs{s},t), \bs{s}\in \cals, t\in\{t_1,\ldots,t_{732}\}\}$. 
This setup is also considered in \citet{buhl1}, Section~5, and \citet{steinkohlphd}, Chapter~7. 
To make the results obtained there comparable to ours, we use the the same preprocessing steps; for a precise description cf. \cite{buhl1}, Section~5.1.   

The data do not fail the max-stability check described in Section~5.2 of \cite{buhl1}, such that we assume that $\{\eta(\bs{s},t), \bs{s}\in \cals, t\in\{t_1,\ldots,t_{732}\}\}$ are realizations of a max-stable space-time process with standard unit Fr{\'e}chet margins.
Nevertheless, the assumption that the data are in fact an exact realization from a max-stable process is only approximate.  Hence there is no guarantee that composite likelihood estimation applied to these transformed data outperforms the semiparametric estimation introduced in Section~\ref{Desmodel}; cf. the results obtained in Section~10 of the supplement~\cite{BDKSsupp}
when data have observational noise. 
Here we use this data example to illustrate our new semiparametric methodology. 

We fit the Brown-Resnick process \eqref{limitfield} by estimating 
\eqref{delta} as follows:
\begin{enumerate}[leftmargin=*]
\item[(1)]
We estimate the parameters $\theta_1$, $\alpha_1$, $\theta_2$ and $\alpha_2$ by WLSE as described in Section~\ref{Desmodel} based on the sets $\calv=\{1,\sqrt{2},2,\sqrt{5},\sqrt{8},3,\sqrt{10},\sqrt{13},4,\sqrt{17}\}$ and $\calu=\{1,\ldots,10\}$. 
Permutation tests as described below and visualized in Figure~\ref{testpermutehour} indicate that these lags are sufficient to cover the relevant extremal dependence structure.
We choose as weights for the different spatial and temporal lags $v \in \calv$ and $u \in \calu$  the corresponding estimated averaged extremogram values; i.e., $w_v=T^{-1}\sum_{k=1}^T \wt{\chi}^{(t_k)}(v,0)$ and $w_u=n^{-2}\sum_{i=1}^{n^2} \wt{\chi}^{(\bs s_i)}(0,u)$, respectively.
 Since the so defined weights are random, what follows is conditional on the realizations of these weights.
 
 As the number of spatial points in the analysis is rather small, we cannot choose a very high empirical quantile $q$, since this would in turn result in a too small number of exceedances to get a reliable estimate of the extremogram. 
Hence, we choose $q$ as the empirical $60\%-$quantile, relying on the fact that the block maxima generate at least approximately a max-stable process and on the robustness of the estimates derived in Section~9 of the supplement~\cite{BDKSsupp}.

For the temporal estimation, we choose the empirical $90\%-$quantile for $q$.
\item[(2)]
We perform subsampling by constructing subsets of the observations and estimating on the subsets (see Section~7 of the supplement~\cite{BDKSsupp}) 
to construct $95\%$-confidence intervals for each parameter estimate. 
As subsample block sizes we choose $b_s=12$ (due to the small number of spatial locations)  for the spatial dimensions and $b_t=300$ for the temporal one. 
As overlap parameters we take $e_s=e_t=1$, which corresponds to the maximum degree of overlap.
\end{enumerate}

The results are shown in Figures~\ref{extrainhour}, \ref{testpermutehour} and Table~\ref{Semirainhour}.
Figure~\ref{rainhour} visualizes the daily rainfall maxima for the two grid locations $(1,1)$ and $(5,6)$.
The semiparametric estimates together with subsampling confidence intervals are given in Table~\ref{Semirainhour}. 

For comparison we present the parameter estimates from the pairwise likelihood estimation (for details see \citet{Steinkohl} and \cite{steinkohlphd}, Chapter~7), where we obtained $\wt{\theta}_1 = 0.3485$, $\wt{\alpha}_1 = 0.8858$, $\wt{\theta}_2=2.4190 $ and $\wt{\alpha}_2=0.1973$. 
From Table~\ref{Semirainhour} we recognize that these estimates are close to the semiparametric estimates and even lie in most cases in the $95\%$-subsampling confidence intervals.

Figure~\ref{extrainhour} shows the temporal and spatial mean of empirical temporal (left) and spatial (right) extremograms as described in \eqref{extrspacemean} and \eqref{extrtimemean} together with 95\% subsampling confidence intervals. 
We perform a permutation test to test the presence of extremal independence. 
To this end we randomly permute the space-time data and calculate empirical extremograms as before. 
More precisely, we compute the empirical temporal extremogram as before and repeat the procedure 1000 times.
From the resulting temporal extremogram sample we determine  nonparametric $97.5\%$ and $2.5\%$ empirical quantiles, which gives a $95\%-$confidence region for temporal extremal independence. 
The analogue procedure is performed for the spatial extremogram. 

The results are shown in Figure~\ref{testpermutehour} together with the extremogram fit based on the WLSE. The plots indicate that for time lags larger than~3 there is no temporal extremal dependence,  and for spatial lags larger than~4 no spatial extremal dependence.

\begin{center}
\captionsetup{type=table}
\begin{tabular}{c|c|c||c|c}
\hline
Estimate & $\wh{\theta}_1$ &  0.3611  & $\wh{\alpha}_1$ & 0.9876  \\
Subsampling-CI & & [0.3472,0.3755] & & [0.9482,1.0267]\\
\hline
Estimate & $\wh{\theta}_2$ & 2.3650 & $\wh{\alpha}_2$ & 0.0818 \\
Subsampling-CI && [1.9110,2.7381] && [0.0000,0.2680]\\
\hline
\end{tabular}
\captionof{table}{Semiparametric estimates for the spatial parameters $\theta_1$ and $\alpha_1$ and the temporal parameters $\theta_2$ and $\alpha_2$ of the Brown-Resnick process in \eqref{limitfield} together with 95\% subsampling confidence intervals.}
\label{Semirainhour}
\end{center}

\begin{figure}[h] 
\centering
\hspace*{-0.3cm}
\subfloat[]{\includegraphics[scale=0.3]{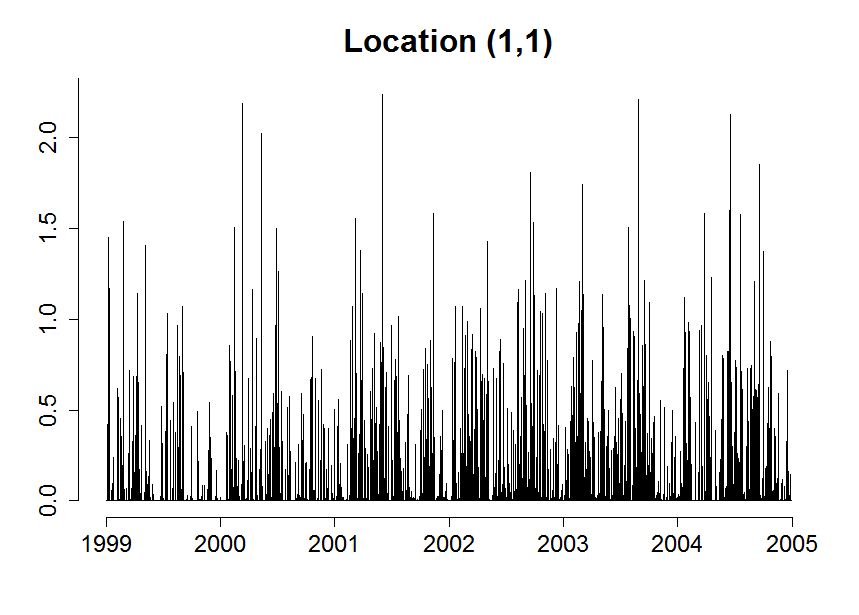}}
\subfloat[]{\includegraphics[scale=0.3]{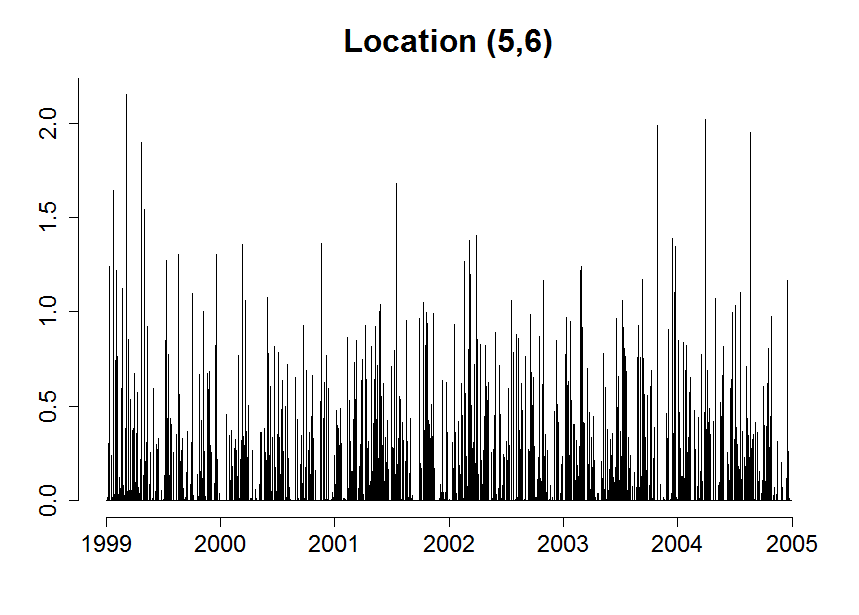}}\\[-7mm]
\caption{Daily rainfall maxima over hourly accumulated measurements from 1999-2004 in inches for two grid locations.}
\label{rainhour}
\end{figure}

\begin{figure}[h] 
\centering
\subfloat[]{\includegraphics[scale=0.22]{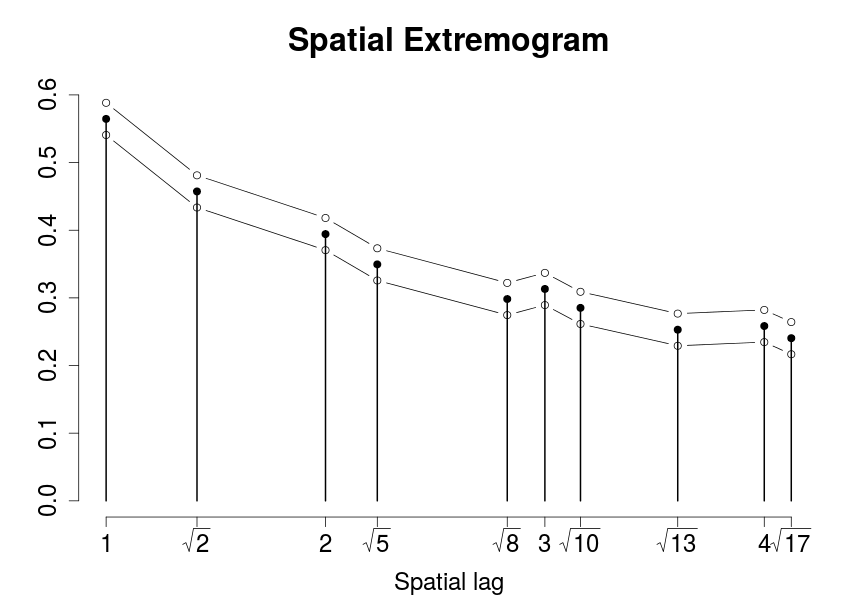}}
\subfloat[]{\includegraphics[scale=0.22]{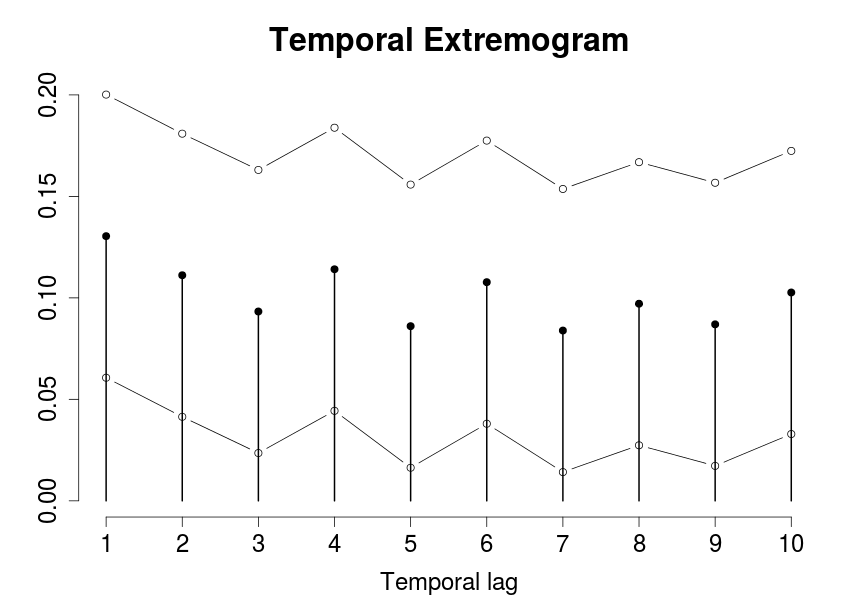}}\\[-5mm]
\caption{Empirical spatial (left) and temporal (right) extremogram based on spatial and temporal means for the space-time observations as given in \eqref{extrspacemean} and \eqref{extrtimemean} together with $95\%-$subsampling confidence intervals.}
\label{extrainhour}
\end{figure}

\begin{figure}[h] 
\centering
\subfloat[]{\includegraphics[scale=0.22]{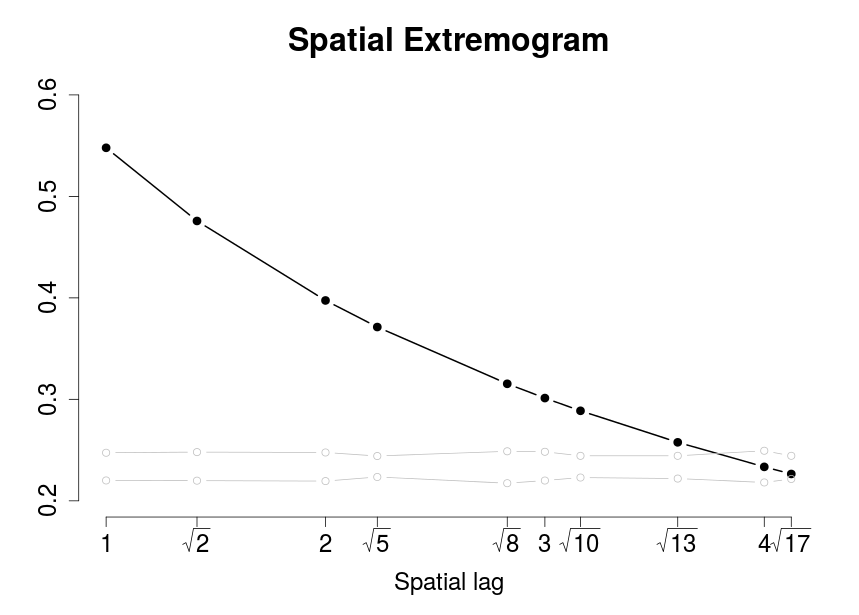}}
\subfloat[]{\includegraphics[scale=0.22]{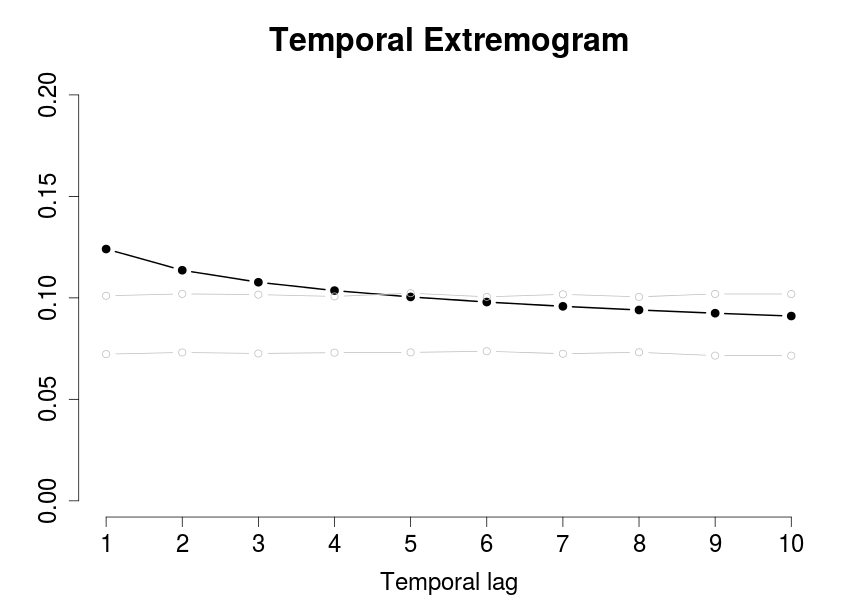}}
\caption{Permutation test for extremal independence: The gray lines show the $97.5\%-$ and $2.5\%-$quantiles of the extremogram estimates for 1000 random space-time permutations for the empirical spatial (left) and the temporal (right) extremogram estimates.}
\label{testpermutehour}
\end{figure}

\section{{Conclusions and Outlook}}

For isotropic strictly stationary regularly-varying space-time processes with additively separable dependence structure we have suggested a new semiparametric estimation method.
The method works remarkably well and produces reliable estimates that are much faster to compute than composite likelihood estimates.
These estimates can also be useful as initial values for a composite likelihood optimization.

Meanwhile, we have generalized the semiparametric method based on extremogram estimation.
The paper \citet{buhl6} is dedicated to the three topics:
\begin{enumerate}
\item
Generalize the dependence function \eqref{delta} to anisotropic and appropriate mixed models and get rid of the assumption of separability.
\item
Generalize the sampling scheme to a fixed (small) number of spatial observations and limit results for the number of temporal observations to tend to infinity.
\item
Generalize the least squares estimation to estimate spatial and temporal parameters simultaneously, also in the situation described in 2.
\end{enumerate}

Another question concerns the optimal choice of the weight matrix $W$, such that the asymptotic variance of the WLSE is minimal. Some ideas can be found in the geostatistics literature in the context of LSE of the variogram parameters; e.g. in \citet{Lahiri2}, Section~4. 
Here the optimal choice of the weight matrix is given by the inverse of the asymptotic covariance matrix of the nonparametric estimates; i.e., of $\big(T^{-1}\sum_{k=1}^T\wt{\chi}^{{(t_k})}(v,0)\big)\trans_{v \in \calv}$ in the spatial case and of $\big(n^{-2}\sum_{i=1}^{n^2} \wt{\chi}^{(\bs{s}_i)}(0,u) - \chi(0,u)\big)\trans_{u \in\calu}$ in the temporal case. 
In our case, however, this involves the matrices $\Pi_2^{(\text{iso})}$ and $\Pi_2^{(\text{time})}$ (given in equations (4.3)-(4.6) of \citet{buhl3}), whose components are infinite sums.

\section*{Acknowledgements}
The three last authors gratefully acknowledge support by the TUM Institute for Advanced Study.
Furthermore, all authors would like to thank Chin Man Mok and Daniel Straub for their help in finding the data and discussions regarding the results. 
We further acknowledge the Southwest Florida Water Management District for providing the data.
We thank Dennis Leber for comparing simulation methods of the BR process, and Ton Dieker and Marius Hofert for improving the simulation code.
SB and CS thank the International Graduate School of Science and Engineering (IGSSE) of the Technical University of Munich for support. The research of RD was supported in part by the National Science Foundation grant DMS-1107031, and ARO MURI grant W11NF-12-1-0385.

\vspace*{-0.35cm}

\begin{supplement}[id=suppA]
\vspace*{-0.25cm}
\stitle{Supplement to ``Semiparametric estimation for isotropic max-stable space-time processes''}
\slink[doi]{COMPLETED BY THE TYPESETTER}
\sdatatype{BDKSsupp.pdf}
\sdescription{We provide additional results on $\alpha-$mixing, subsampling for confidence regions, and a simulation study supporting the theoretical results. Our method is extended to max-stable date with observational noise and applied to both exact realizations of the Brown-Resnick process and to realizations with observational noise, thus verifying the robustness of our approach.}
\end{supplement}

\vspace*{-0.35cm}

\bibliographystyle{plainnat}
\bibliography{bibtex_spacetime}

\end{document}


\begin{frontmatter}
\title{Supplement to the paper ``Semiparametric estimation for isotropic max-stable space-time processes''}
\runtitle{Semiparametric estimation for isotropic max-stable space-time processes: Supplement}

\begin{aug}
\author{\fnms{Sven} \snm{Buhl}\thanksref{a,e1}\ead[label=e1,mark]{sven.buhl@tum.de}}
\author{\fnms{Richard A.} \snm{Davis}\thanksref{b,e2}\ead[label=e2,mark]{rdavis@stat.columbia.edu}%
\ead[label=u2,url]{http://www.stat.columbia.edu/\textasciitilde rdavis}}
\author{\fnms{Claudia} \snm{Kl{\"u}ppelberg}\thanksref{a,e3}\ead[label=e3,mark]{cklu@tum.de}}
\and
\author{\fnms{Christina} \snm{Steinkohl}\thanksref{a,e4}%
\ead[label=e4,mark]{christina.steinkohl@gmail.com}%
\ead[label=u1,url]{http://www.statistics.ma.tum.de}}

\address[a]{Center for Mathematical Sciences and TUM Institute of Advanced Study, Technische Universit{\"a}t M{\"u}nchen, Boltzmannstr. 3, 85748 Garching, Germany.
\printead{e1,e3,e4},
\printead{u1}}

\address[b]{Department of Statistics, Columbia University, 1255 Amsterdam Avenue, New York, NY 10027, USA.
\printead{e2},
\printead{u2}}

\runauthor{S. Buhl et al.}

\begin{keyword}
\kwd{Brown-Resnick process}
\kwd{extremogram}
\kwd{max-stable process}
\kwd{regular variation}
\kwd{semiparametric estimation}
\kwd{space-time process}
\kwd{subsampling}
\kwd{mixing}
\end{keyword}

\end{aug}

\end{frontmatter}


\noindent
This supplementary material provides additional definitions  and results to the paper \cite{Steinkohl3}, where the setting, notation,  equation reference numbers are retained from that paper.
Section~\ref{theoryspatial1} defines $\alpha$-mixing and states results for Brown-Resnick space-time processes used in the proof of Theorem~4.3 and throughout this supplement. 
Within the particular space-time setting considered in the paper, we provide insight into subsampling to obtain adequate confidence regions for the true parameters in Section~\ref{Sec:subsampling}. 
Section~\ref{appendB} states and proves an important result related to the extremogram for the Brown-Resnick process observed with noise.  This result provides the theoretical justification for the robustness of WLSE for space-time data based on small departures from the Brown-Resnick model.  
The simulation study presented in Section~\ref{Simulation} confirms these results and other findings of the paper.

\setcounter{section}{6}

\section{{Subsampling for confidence regions}} \label{Sec:subsampling}

As in Sections~2 and 3 of the paper \cite{Steinkohl3} we consider a strictly stationary regularly varying process in space and time $\{\eta(\bs s,t): \bs s \in \mathbb{R}^{d-1}, t \in [0,\infty)\}$ for $d\in\N$.
We assume additively separable parametric models for its extremogram $\{\chi(v,u), v,u \geq 0\}$, such that setting either the temporal lag $u$ or the spatial lag $v$ equal to $0$, it can be linearly parametrized as 
\begin{align*}
T_1(\chi(v,0)) &=T_1(\chi(v,0;C_1,\alpha_1))=C_1+\alpha_1 v, \quad (C_1,\alpha_1) \in \Theta_{\mathcal{S}}, \quad v \geq 0, \\
 T_2(\chi(0,u)) &=T_2(\chi(0,u;C_2,\alpha_2))=C_2+\alpha_2 u, \quad (C_2,\alpha_2) \in \Theta_{\mathcal{T}} \quad u\geq 0,
\end{align*}
  where $T_1$ and $T_2$ are known suitable strictly monotonous continuously differentiable transformations and the parameters $(C_1,\alpha_1)$ and $(C_2,\alpha_2)$ lie in appropriate parameter spaces $\Theta_{\mathcal{S}}$ and $\Theta_{\mathcal{T}}$.

The estimation method  described in Sections~2 and 3, which is based on the (averaged) empirical extremogram computed by means of space-time observations on the grid $\cals_n \times \{t_1,\ldots,t_T\}$ defined in Condition~2.3, 
yields a consistent and asymptotically normal estimator $\wh{\bs\psi}_1=(\wh C_1,\wh\alpha_1)\trans$ of the true parameter vector $\bs\psi_1^\star=(C_1^\star,\alpha_1^\star)\trans$. The rate of convergence is given by $\tau_n:=n/m_n$, where $m_n$ is an appropriately chosen scaling sequence.

Due to the complicated forms of the covariance matrix of the normal limit distribution (cf. Theorem~3.1 
and Theorem~3.19 of \cite{buhlphd}) we use resampling methods to construct asymptotic confidence regions for $\bs\psi_1^\star$. One appealing method is subsampling (see \citet{Politis4}, Chapter~5), since it works under weak regularity conditions and produces asymptotically correct coverage.
The central assumption is the existence of a continuous weak limit law, which is guaranteed by Theorem~3.1. 
Again we only consider the spatial case, the temporal case is described (again for the example of the Brown-Resnick process) in Section~3.4.2 of \cite{buhlphd}.

We have applied subsampling successfully already for confidence bounds of pairwise likelihood estimates of the max-stable space-time Brown-Resnick process in \citet{buhl1}, Section~4. 
The procedure is as follows: {understanding inequalities between vectors componentwise,} we choose block lengths $\bs{b}=(b_s,b_s,T)$ with $(1,1) \leq (b_s,b_s) \leq (n,n)$ and the degree of overlap $\bs{e}=(e_s,e_s,T)$ with $(1,1)\leq (e_s,e_s) \leq (b_s,b_s)$, where $\bs e=(1,1,T)$ corresponds to maximum overlap and $\bs e=\bs b$ to no overlap. 
The blocks are indexed by $\bs i=(i_1,i_2) \in \mathbb{N}^2$ with $i_j \leq q_s$ for $q_s:=\lfloor \frac{n-b_s}{e_s}\rfloor + 1$ and $j=1,2$. 
This results in a total number of $q= q_s^2$ blocks, which we summarize in the sets
\begin{align*} 
E_{\bs{i},\bs{b},\bs{e}}=\big\{(s_1, s_2) \in \cals_n: &(i_j-1)e_s+1 \leq s_j \leq (i_j-1)e_s+b_s\mbox{ for } j=1,2\big\}\times \{t_1,\ldots,t_T\}.
\end{align*}
We estimate the parameters based on the observations in each block as described in the previous sections. 
This yields 
different estimates, which we denote by
$\wh{\bs\psi}_{1,\bs{i}}$.


Theorem~\ref{subsampling_conf} below provides a basis for constructing asymptotically valid confidence intervals for the true parameters $C_1^\star$ and $\alpha_1^\star$. We define $\tau_{b_s}=b_s/m_{b_s}$ as the analogue of $\tau_n=n/m_n$.

\begin{theorem}\label{subsampling_conf}
Assume that the conditions of Theorem~3.4 hold, and 
\begin{enumerate}[label=(\roman*)]  
\item $b_s \rightarrow \infty$ such that
$b_s=o({n})$ and $\tau_{b_s}/\tau_n \rightarrow 0$ as $\nto$
\item $\bs{e}$ does not depend on $n$,
\item the $\alpha$-mixing coefficients $\alpha_{k,\ell}(\cdot)$ defined in \eqref{alphaBolt} satisfy
\begin{align*}
\frac{1}{n^2} \sum_{r=1}^n r \alpha_{b,b}(r) \rightarrow 0, \quad  \nto,
\end{align*} 
where $b:=b_s^2$.
\end{enumerate}
Define the empirical distribution function {$L_{b_s,s}$}
\begin{align*}
L_{b_s,s}(x):=\frac{1}{q} \sum\limits_{i_1=1}^{q_s} \sum\limits_{i_2=1}^{q_s} \mathds{1}_{\left\{\tau_{b_s}\left\|\wh{\bs\psi}_{1,\bs{i}}-\wh{\bs\psi}_{1}\right\| \leq x\right\}},\quad x\in\R,
\end{align*}
and the empirical quantile function
\begin{align*} 
c_{b_s,s}(1-\alpha):=\inf\left\{x\in\R :L_{b_s,s}(x) \geq 1-\alpha\right\},\quad \alpha\in(0,1).
\end{align*} 
Then \begin{align}
\mathbb{P}\left(\tau_n\|\wh{\bs\psi}_{1}-{\bs\psi}_{1}^\star\| \leq c_{b_s,s}(1-\alpha) \right) \rightarrow 1-\alpha,\quad \nto. \label{subs_asy_ci_space}
\end{align}
\end{theorem}

\bproof
We apply Corollary~5.3.3 of \citet{Politis4}. Their main  Assumption~5.3.3 is the existence of a continuous weak limit distribution of $\tau_n \|\wh{\bs\psi}_{1}-{\bs\psi}_{1}^\star\|$, which holds by Theorem~3.1. 
The remaining assumptions (i)-(iii) are also presumed in \citet{Politis4}. 
\eproof 

As a consequence of equation~\eqref{subs_asy_ci_space}, for $n$ large enough, an approximate $(1-\alpha)$-confidence region for the true parameter vector $\bs\psi_1^\star=(C_1^\star,\alpha_1^\star)$ is given by 
\begin{align}
\{\bs \psi \in \Theta_{\cals}: \|\bs \psi - \wh{\bs \psi}_1 \| \leq c_{b_s,s}(1-\alpha)/\tau_n\}, \label{Subsampl_CI}
\end{align}
where $\Theta_{\cals}$ denotes as before the parameter space.

\brem
Consider the special case of the Brown-Resnick process~(4.1) 
with dependence function $\delta$ given in~(4.2) as
\begin{align}\label{delta}
\delta(v,u) = 2\theta_1 v^{\alpha_1}+2\theta_2u^{\alpha_2}, \quad v,u \geq 0, \quad \theta_1,\theta_2>0,\quad 0<\alpha_1,\alpha_2\leq 2,
\end{align} 
whose parameters $\alpha_1$ and $\alpha_2$ have bounded support. 
Recall from Section~4 that to put this in the context of this section, we set 
$$C_1=\log(\theta_1)\quad\mbox{and}\quad C_2=\log(\theta_2)$$ 
and choose the transformations $T_1$ and $T_2$ defined by 
$$T_1(\chi(v,0))=2\log\big(\Phi^{-1}\big(1-\frac{1}{2}\chi(v,0)\big)\big)\quad\mbox{and}\quad T_2(\chi(0,u))=2\log\big(\Phi^{-1}\big(1-\frac{1}{2}\chi(0,u)\big)\big).$$ 
In the following, we focus on the spatial parameters. The parameter space is given by $\Theta_{\cals}=\mathbb{R} \times (0,2]$. 
Since the parameter space for $\alpha$ is bounded, we use the constrained estimate $\wh{\bs\psi}_{1}^c$ defined before Theorem~4.6. We denote the true parameter by ${\bs\psi}_{1}^\star=(\log(\theta_1^\star),\alpha_1^\star)$. The assumptions of Theorem~\ref{subsampling_conf} for subsampling are satisfied in this setting. Particularly important is the existence of a continuous weak limit distribution of $\tau_n \|\wh{\bs\psi}_{1}^c-{\bs\psi}_{1}^\star\|$, where the scaling sequence is given by $\tau_n=n/m_n$.
By Theorem~4.6,
the continuous mapping theorem and the Fubini-Tonelli theorem we have for $\ga \geq 0$, as $\nto$,
\begin{align*}
\mathbb{P}(\tau_n \|\wh{\bs\psi}_{1}^c-{\bs\psi}_{1}^\star\| \leq \ga) \rightarrow \mathbb{P}(\|\bs Z_1 \| \leq \ga)&=\mathbb{P}(\bs Z_1 \in B(\bs 0,\ga))
= 2\int\limits_{-\ga}^\ga \int\limits_{0}^{\sqrt{\ga^2-r^2}} \varphi_{\bs 0, \Pi_3^{(\textnormal{iso})}}(r,s)dsdr
\end{align*}
if $\alpha_1^\star<2$. For  $\alpha_1^\star=2$ we obtain
\begin{align*}
&\mathbb{P}(\tau_n  \|\wh{\bs\psi}_{1}^c-{\bs\psi}_{1}^\star\| \leq \ga) 
\, \to\, \mathbb{P}(\|\bs Z_2 \| \leq \ga) =\mathbb{P}(\bs Z_2 \in B(\bs 0,\ga))\\
= &\int\limits_{-\ga}^\ga\int\limits_{-\sqrt{\ga^2-r^2}}^0  \varphi_{\bs 0, \Pi^{(\textnormal{iso})}_3}(r,s)dsdr \\
& + \int\limits_{-\ga}^\ga  \int\limits_0^{\infty} \varphi_{\bs 0, \Pi^{(\textnormal{iso})}_3}(r-\frac1{\sum_{v \in \calv} w_v}\sum_{v \in \calv} (w_vx_v)s,s) dsdr \\
= & \int\limits_{-\ga}^\ga \Bigg\{ \int\limits_{-\sqrt{\ga^2-r^2}}^0  \varphi_{\bs 0, \Pi^{(\textnormal{iso})}_3}(r,s) ds + \int\limits_0^{\infty} \varphi_{\bs 0, \Pi^{(\textnormal{iso})}_3}\Big(r-\frac1{\sum_{v \in \calv} w_v}\sum_{v \in \calv} (w_vx_v)s,s\Big)ds\Bigg\}dr .
\end{align*}
In particular, the limiting  distribution function of the scaled norm $\tau_n  \|\wh{\bs\psi}_{1}^c-{\bs\psi}_{1}^\star\|$ is continuous in $\ga$ both for $\alpha_1^\star<2$ and $\alpha_1^\star=2$.

The required condition (iii) on the $\alpha$-mixing coefficients is satisfied, similarly as in the proof of Theorem~3.1, 
by equation~\eqref{mixingspace} below.

As in~\eqref{Subsampl_CI}, for $n$ large enough, an approximate $(1-\alpha)$-confidence region for the true parameter vector $\bs\psi_1^\star=(\log(\theta_1^\star),\alpha_1^\star)$ is given by 
$$\{\bs \psi \in \mathbb{R} \times (0,2]: \|\bs \psi - \wh{\bs \psi}_1^c \| \leq c_{b_s,s}(1-\alpha)/\tau_n\}.$$
The one-dimensional approximate $(1-\alpha)$-confidence intervals for the parameters $\theta_1^\star$ and $\alpha_1^\star$ can be read off from this as
\begin{align*}
\Big[\wh\theta_1^c \exp\Big\{-\frac{c_{b_s,s}(1-\alpha)}{\tau_n}\Big\},\wh\theta_1^c \exp\Big\{\frac{c_{b_s,s}(1-\alpha)}{\tau_n}\Big\}\Big] \text{ and}\\
\Big[\wh\alpha_1^c -\frac{c_{b_s,s}(1-\alpha)}{\tau_n},\wh\alpha_1^c +\frac{c_{b_s,s}(1-\alpha)}{\tau_n}\Big] \cap (0,2].
\end{align*}
\erem

\section{$\alpha$-mixing of the Brown-Resnick space-time process}\label{theoryspatial1}

We define $\alpha$-mixing for spatial processes; see e.g.  Doukhan \cite{Doukhan} or Bolthausen \cite{Bolthausen}.

\begin{definition}\label{mixing}
For $d \in \bbn$, consider a strictly stationary process $\left\{X(\bs{s}): \bs{s}\in \bbr^d\right\}$ and let $d(\cdot,\cdot)$ be some metric induced by a norm on $\mathbb{R}^d$.
For $\Lambda_1, \Lambda_2 \subset \mathbb{Z}^d$ set 
\begin{align*}
d(\Lambda_1,\Lambda_2) := \inf\left\{d(\bs{s}_1, \bs{s}_2): \ \bs{s}_1 \in \Lambda_1, \bs{s}_2 \in \Lambda_2\right\}.
\end{align*}
Further, for $i=1,2$ denote by $\mathcal{F}_{\Lambda_i}= \sigma\left\{X(\bs{s}), \bs{s}\in \Lambda_i\right\}$ the $\sigma$-algebra generated by $\{X(\bs{s}): \ \bs{s}\in \Lambda_i\}$.
\begin{enumerate}[label=(\roman*)]
\item 
The {\em $\alpha$-mixing coefficients} are defined for $k,l \in \bbn \cup \{\infty\}$ and $r \geq 0$ by
\begin{equation}\label{alphaBolt}
\hspace*{-0.7cm}
\alpha_{k,l}(r) = \sup\left\{\left|\mathbb{P}(A_1\cap A_2) - \mathbb{P}(A_1)\mathbb{P}(A_2)\right|: \ A_i \in \mathcal{F}_{\Lambda_i}, |\Lambda_1|\leq k, |\Lambda_2|\leq l, d(\Lambda_1,\Lambda_2) \geq r\right\},
\end{equation}
where $|\Lambda_i|$ is the cardinality of the set $\Lambda_i$ for $i=1,2$.
\item 
The random field is called {\em $\alpha$-mixing}, if $\alpha_{k,l}(r) \to 0$ as $r\to\infty$ for all $k,l\in \bbn$.
\end{enumerate}
\end{definition}

For a strictly stationary max-stable process Corollary~2.2 of Dombry and Eyi-Minko~\cite{Dombry} shows that the $\alpha$-mixing coefficients can be related to the extremogram of the max-stable process. Equations~\eqref{mixingspace} and~\eqref{mixingtime} follow as in the proofs of Proposition~1 and~2 of \citet{buhl1}. 

\begin{proposition}
For all fixed time points $t\in \bbn$ the random field $\left\{\eta(\bs{s},t), \bs{s}\in \bbz^2\right\}$ defined by (4.1) 
is $\alpha$-mixing with mixing coefficients satisfying
\begin{equation}\label{mixingspace}
\alpha_{k,l}(r) \leq 2kl\sup_{s\geq r}\chi(s,0) \leq 4kl e^{-\theta_1 r^{\alpha_1}/2},  \quad k,l \in \bbn, \, r\geq 0.
\end{equation}
For all fixed locations $\bs{s}\in \mathbb{R}^2$ the time series $\left\{\eta(\bs{s},t): t\in [0,\infty)\right\}$ in (4.1) 
is $\alpha$-mixing with mixing coefficients satisfying for some constant $c>0$
\begin{equation}\label{mixingtime}
\alpha(r):=\alpha_{\infty,\infty}(r) \leq c\sum_{u= r}^{\infty} u e^{-\theta_2u^{\alpha_2}/2 }, \quad r\geq 0.
\end{equation}
\end{proposition}

We will make frequent use of the following simple result.

\begin{lemma}\label{remarkalpha}
Let $z\in\N$.
For $(\theta,\alpha) \in \{ (\theta_1,\alpha_1),(\theta_2,\alpha_2)\}$ and sufficiently large $r$ such that the sequence $u^z e^{-\theta u^{\alpha}/2}$ is  decreasing for $u \geq r$, we have
$$g_z(r) = \sum_{u=r}^{\infty}u^z e^{-\theta u^{\alpha}/2}
\leq ce^{-\theta r^{\alpha}/2} r^{z+1},\quad r\in\N.$$
for some constant $c=c(z)>0$.
\end{lemma}

\begin{proof}
An integral bound together with a change of variables yields
\begin{align*}
g_z(r) &= r^ze^{-\theta r^{\alpha}/2}+\sum_{u=r+1}^{\infty}u^ze^{-\theta u^{\alpha}/2} \leq r^ze^{-\theta r^{\alpha}/2}+ \int_{r}^{\infty}u^ze^{-\theta u^{\alpha}/2} du\\
 &= r^ze^{-\theta r^{\alpha}/2}+ \left(\frac{2}{\theta}\right)^{(z+1)/\alpha }\frac{1}{\alpha}\int_{\theta r^{\alpha}/2}^{\infty}t^{(z+1)/\alpha-1}e^{-t}dt \\
& \leq r^ze^{-\theta r^{\alpha}/2}+ c_1\Gamma\left(\left\lceil (z+1)/\alpha\right\rceil,\theta r^{\alpha}/2\right) \\
 & = r^ze^{-\theta r^{\alpha}/2}+ c_1\left(\left\lceil (z+1)/\alpha\right\rceil-1\right)!\ e^{-\theta r^{\alpha}/2}\sum_{k=0}^{\left\lceil (z+1)/\alpha\right\rceil-1}\frac{\theta^k r^{\alpha k}}{2^k k!} \\
&\leq  r^ze^{-\theta r^{\alpha}/2}+c_2e^{-\theta r^{\alpha}/2} r^{\alpha (\lceil (z+1)/\alpha\rceil-1)} \\
&\leq ce^{-\theta r^{\alpha}/2}r^{z+1},
\end{align*}
where $\Gamma(s,r) = \int_r^{\infty}t^{s-1}e^{-t} dt = (s-1)!e^{-r}\sum_{k=0}^{s-1}r^k/k!$, $s \in \mathbb{N}$, is the incomplete gamma function and $c_1,c>0$ are constants depending on $z$. 
\end{proof}

\section{Robustness of the bias corrected estimator} \label{appendB}

As shown in the simulation study in Section~\ref{Simulation} below, the WLSEs are robust with respect to small deviations from the model assumptions. Specifically, if one adds measurement noise to the underlying Brown-Resnick process, the WLSEs still perform well.  This is in contrast to the composite likelihood procedure for which the estimates become biased.  The theoretical foundation for the good performance of the WLSEs is given in Lemma~\ref{le3.2ii}, which is the analogue of Lemma~4.2 for the Brown-Resnick process without noise.
\ble\label{le3.2ii}
Let $\{Z(\bs s,t): (\bs s,t) \in \bbr^2\times [0,\infty)\}$ be i.i.d. random variables which are independent of the space-time process $\{\eta(\bs s,t): (\bs s,t) \in \bbr^2\times [0,\infty)\}$.
Assume the moment condition $E|Z(\bs 0,0)|^{2+\epsilon} < \infty$ for some $\epsilon>0$.
Then for every sequence $a_n \to \infty$ we have for fixed $t \in [0,\infty),$
\begin{align*}
&\frac{\mathbb{P}(\eta(\bs{s},t)+Z(\bs{s},t)>a_n,\eta(\bs{s}+\bs{h},t)+Z(\bs{s}+\bs{h},t)>a_n)}{\mathbb{P}(\eta(\bs{s},t)+Z(\bs{s},t)>a_n)} \\
 =&\chi(\|\bs h \|,0)+ \Big[\frac{1}{2a_n}\big(\chi(\|\bs h \|,0)-2\big)\big(\chi(\|\bs h \|,0)-1\big)\Big](1+o(1)).
\end{align*}
\ele

\begin{proof}
For notational simplicity,  write $\eta_1=\eta(\bs s,t)$, $\eta_2=\eta_(\bs s+\bs h,t),Z_1=Z(\bs s,t), Z_2=Z(\bs s+\bs h,t)$, and $\chi=\chi(\|\bs h\|)$.  Here we assume that $\bs h\neq  \bs 0$, since otherwise, $\chi=1$ and thus 
$$
0=(\chi-2)(\chi-1)=\Bigg(\frac{\mathbb{P}(\eta_1+Z_1 >  a_n)}{\mathbb{P}(\eta_1+Z_1 > a_n)} - \chi\Bigg)\,.
$$
Using (4.3) 
and the independence of  $(\eta_1,\eta_2)$ with $(Z_1,Z_2)$, we have
\beao
&&\mathbb{P}(\eta_1+Z_1>a_n,\eta_2+Z_2>a_n)\\
&=&  1- \mathbb{P}(\eta_1+Z_1\le a_n)-\mathbb{P}(\eta_2+Z_2\le a_n)
 +\mathbb{P}(\eta_1+Z_1\le a_n,\eta_2+Z_2\le a_n)\\
&=&2\mathbb{P}(\eta_1+Z_1> a_n) -(1-\mathbb{P}(\eta_1+Z_1\le a_n,\eta_2+Z_2\le a_n))\\
&=&2\mathbb{P}(\eta_1+Z_1> a_n)-\E\Big[1-\exp\Big\{-\frac{1}{(a_n-Z_1)_+}\Phi_{1,2}^n-\frac{1}{(a_n-Z_2)_+}\Phi_{2,1}^n\Big\}\Big]\,,
\eeao
where $x_+=\max\{0,x\}$, $\Phi_{i,j}^n=\Phi(c \log((a_n-Z_j)_+/(a_n-Z_i)_+) +c/2)$, and $c=\sqrt{2\delta(\|h\|)}$. Set $\Phi^*=\lim_{n\to\infty}\Phi_{i,j}^n=\Phi(c/2)~a.s.$ 

Take $b_n=a_n^{1-\epsilon/4}$, where $\epsilon \in (0,1)$ is specified in the statement of the lemma.  Then it follows that $b_n/a_n\to 0$, $a_n^2/b_n^{2+\epsilon}=a_n^{-(2-\epsilon)\epsilon/4} \to 0$ and hence
\beao 
\mathbb{P}(|Z|\ge b_n)\le \frac{\E|Z|^{2+\epsilon}}{b_n^{2+\epsilon}}=o(a_n^{-2})\,.
\eeao
Writing $\E_n$ for expectation relative to the restriction on the event $\{|Z_1|\vee |Z_2|\le b_n\}$, we have for any bounded sequence of random variables $Y_n$ that $a_n^2(\E Y_n-\E_nY_n)\to 0$. Hence, using a Taylor series approximation, we obtain
\beam\label{etaplusz}
\mathbb{P}(\eta_1+Z_1>a_n)&=&\E\Big[1-\exp\Big\{-\frac{1}{(a_n-Z_1)_+}\Big\}\Big] \nonumber\\
&=&\E_n\frac{1}{(a_n-Z_1)_+}-\E_n\frac{1}{2(a_n-Z_2)^2} +o(a_n^{-2})\nonumber\\
&=&\E_n\frac{1}{(a_n-Z_1)_+}-a_n^{-2}+o(a_n^{-2})
\eeam
and
\beao
I_1 &:=&\E\Big[1-\exp\Big\{-\frac{1}{(a_n-Z_1)_+}\Phi_{1,2}^n-\frac{1}{(a_n-Z_2)_+}\Phi_{2,1}^n\Big\}\Big]\\
&=&\E_n\Big[\frac{1}{(a_n-Z_1)_+}\Phi_{1,2}^n+\frac{1}{(a_n-Z_2)_+}\Phi_{2,1}^n\Big]\\
&& \quad\quad\quad\quad -\frac12
\E_n\Big[\frac{1}{(a_n-Z_1)_+}\Phi_{1,2}^n+\frac{1}{(a_n-Z_2)_+}\Phi_{2,1}^n\Big]^2 +o(a_n^{-2})\\
&=&2\E_n\Big[\frac{1}{(a_n-Z_1)_+}\Phi_{1,2}^n\Big]-\E_n\Big[\frac{1}{(a_n-Z_1)_+}\Phi_{1,2}^n\Big]^2\\
&& \quad\quad\quad\quad -\E_n\Big[\frac{1}{(a_n-Z_1)_+}
\frac{1}{(a_n-Z_2)_+}\Phi_{1,2}^n\Phi_{2,1}^n\Big] +o(a_n^{-2}).
\eeao
In order to complete the proof it suffices to show the following two relations:
\beam\label{relation1}
2\E_n\Big[\frac{1}{(a_n-Z_1)_+}\Phi_{1,2}^n\Big]-2\Phi^*\mathbb{P}(\eta_1+Z_1>a_n) + \frac{\Phi^*}{a_n^{2}} = o(a_n^{-2}),
\eeam
\beam\label{relation2}
\E_n\Big[\frac{1}{(a_n-Z_1)_+}\Phi_{1,2}^n\Big]^2 &+&\E_n\Big[\frac{1}{(a_n-Z_1)_+}
\frac{1}{(a_n-Z_2)_+}\Phi_{1,2}^n\Phi_{2,1}^n\Big]\nonumber\\
&=& 2\frac{(\Phi^*)^2}{a_n^{2}}+o(a_n^{-2}).
\eeam
To see that \eqref{relation1} and \eqref{relation2} yield the result in Lemma \ref{le3.2ii}, observe that
\beao
\frac{\mathbb{P}(\eta_1+Z_1>a_n,\eta_2+Z_2>a_n)}{\mathbb{P}(\eta_1+Z_1> a_n)}&=&
2-\frac{I_1}{\mathbb{P}(\eta_1+Z_1> a_n)}\\
	&=& 2-2\Phi^*+\frac{\Phi^*-2(\Phi^*)^2}{a_n^2\mathbb{P}(\eta_1+Z_1> a_n)}+o_(a_n^{-1})\\
	&=&\chi +\frac12 a_n^{-1} (\chi-2)(\chi-1)+o(a_n^{-1})\,,
\eeao
where we have used the properties $\chi=2-2\Phi^*$ and $a_n\mathbb{P}(\eta_1+Z_1>a_n)\to 1.$.

Using a Taylor series expansion and the relation 
$$
\log\frac{(a_n-Z_2)_+}{(a_n-Z_1)_+}= \frac{Z_1-Z_2}{a_n}+(Z_1-Z_2)^2\cdot O(a_n^{-2}),
$$ 
it follows that on the set $\{|Z_1|\vee |Z_2|\le b_n\}$,
\beao
\Phi_{1,2}^n-\Phi^*&=&c\log\frac{(a_n-Z_2)_+}{(a_n-Z_1)_+}\Phi'(c/2)+O_p(a_n^{-2})\\
&=&c\frac{Z_1-Z_2}{a_n}\Phi'(c/2) +O_p(a_n^{-2})\,.
\eeao
Finally turning to \eqref{relation1} and applying \eqref{etaplusz}, the left-hand side is equal to 
\beao
&&2\E_n\Big[\frac{1}{(a_n-Z_1)_+}(\Phi_{1,2}^n-\Phi^*)\Big]+\frac{\Phi^*}{a_n^2} - o(a_n^{-2})\\
&=&c\E_n\frac{Z_1-Z_2}{a_n(a_n-Z_1)}+ \Phi'(c/2)\frac{\Phi^*}{a_n^2}-o(a_n^{-2})\,,
\eeao
which by multiplying by $a_n^2$ and taking limits gives the desired limit of $\Phi^*$.

Finally, \eqref{relation2} is obtained by multiplying both sides of the equation $a_n^{2}$ and taking limits, where the interchange of limits and expectation are justified by the dominating convergence theorem.  This completes the proof. 
\end{proof}

\section{Simulation study}\label{Simulation}

We examine the performance of the {WLSEs} by simulating a large number of Brown-Resnick processes with dependence function~\eqref{delta}.
Many real data may not follow a Brown-Resnick process precisely. For a more realistic setting, we thus do not perform the simulation study for a Brown-Resnick process only, but also for the sum of a Brown-Resnick process and some noise. This sum is regularly varying and possesses the same dependence function and extremogram for which all results of Sections~2 
and~3 
hold.
As we want to use the bias reduction procedure of Section~4, 
we have to make sure that Lemma~4.2 
extends to this setting; all other results of this section follow from that. 
This is guaranteed by Lemma~\ref{le3.2ii}.

The estimation of the spatial parameters relies on a rather large number of spatial observations and the estimation of the temporal parameters on a rather large number of observed time points. However, simulation of Brown-Resnick space-time processes based on the exact method proposed by \citet{Dombry2} can be time consuming, if both a large number of spatial locations and of time points is taken. 
For a time-saving method we generate the process on two different space-time observation areas, one for examining the performance of the spatial estimates and one for the temporal estimates, which we call $\mathcal{S}^{(1)} \times \mathcal{T}^{(1)}$ and $\mathcal{S}^{(2)} \times \mathcal{T}^{(2)}$, respectively.
The design for the simulation experiment is given in more detail as follows:
\begin{enumerate}[leftmargin=*]
\item We choose two space-time observation areas
\begin{align*}
\mathcal{S}^{(1)} \times \mathcal{T}^{(1)} &= \left\{(i_1,i_2): i_1,i_2 \in \left\{1,\ldots,70\right\}\right\} \times \{1,\ldots,10\}\\
\mathcal{S}^{(2)} \times \mathcal{T}^{(2)} &= \left\{(i_1,i_2): i_1,i_2 \in \left\{1,\ldots,5\right\}\right\} \times \{1,\ldots,300\}
\end{align*}
{and the sets $\calv=\{1,\sqrt{2},2,\sqrt{5},\sqrt{8},3,\sqrt{10},\sqrt{13},4,\sqrt{17}\}$ and $\calu=\{1,\ldots,10\}$.}

\item 
We simulate the Brown-Resnick space-time process (4.1) 
based on the exact method proposed in \citet{Dombry2}, using the \texttt{R}-package \texttt{RandomFields}~\cite{Schlather5}. 
The dependence function $\delta$ is modelled as in~(4.2) (cf. \eqref{delta}); i.e.,
$$\delta(v,u)=2\theta_1 v^{\alpha_1}+2\theta_2u^{\alpha_2}, \quad v,u \geq 0,$$
with parameters 
$$\theta_1 = 0.4, \ \alpha_1=1.5, \quad \theta_2=0.2, \ \alpha_2=1.$$
\item 
The parameters $\theta_1,\alpha_1,\theta_2$ and $\alpha_2$ are estimated.
\begin{itemize}[leftmargin=*]
\item 
For the estimation of the empirical extremograms (cf. equations~(2.5)-(2.8)) 
we have to choose high empirical quantiles $q$. 
In practice, $q$ is chosen from an interval of high quantiles for which the empirical extremogram is robust, see the remarks of \citet{DMZ} after Theorem~2.1. 
We choose the $90\%-$empirical quantile for the estimation of the spatial parameters and the $70\%-$quantile for the temporal part. 
The quantile for the temporal part is lower to ensure reliable estimation of the extremogram, because the number of time points ($300$) used for the estimation of the temporal parameters is much smaller than the number of spatial locations ($70 \cdot 70=4900$)  used for the estimation of the spatial parameters. 
\item 
The weights in the constrained weighted linear regression problem (see (2.9) 
and (2.10)) 
are chosen such that locations and time points which are further apart of each other have less influence on the estimation.
More precisely, we choose
$$w_u = \exp\{-u^2\} \, \mbox{ for } \,  u\in \mathcal{U} \quad\text{ and } \quad w_v = \exp\{-v^2\} \, \mbox{ for } v \in \calv.$$ 
\end{itemize}
{This choice of weights reflects the exponential decay of $\chi(v,0)$ and $\chi(0,u)$ given  in (4.4), 
which are tail probabilities of the standard normal distribution $\Phi$.}
\item 
Pointwise confidence bounds are computed by subsampling as described in Section~\ref{Sec:subsampling} for the spatial parameters of general regularly varying processes and in Section~3.4.2 of \cite{buhlphd} for the temporal parameters of the Brown-Resnick space-time process considered in this section. We choose block lengths $\bs b=(50,50,10)$ and overlap $\bs e=(2,2,10)$ for the space-time process with observation area $\mathcal{S}^{(1)} \times \mathcal{T}^{(1)}$ and $\bs b=(5,5,200)$, $\bs e=(5,5,1)$ for the process with observation area $\mathcal{S}^{(2)} \times \mathcal{T}^{(2)}$.
\item
{Steps 1 - 5 are repeated 100 times.}
\end{enumerate}

Figure \ref{Fig1} shows the WLSEs of the spatial parameters $\theta_1$ and $\alpha_1$ for each of the 100 realizations of the Brown-Resnick space-time process. 
The dashed lines above and below the dots are pointwise confidence intervals based on subsampling. 
Panel (a) of Table~\ref{summarysemi} shows the mean, root mean squared error (RMSE) and mean absolute error (MAE) of both the spatial and the temporal WLSEs based on the 100 simulations.
Altogether, we observe that the estimates are close to the true values.
The spatial estimates are slightly superior to the temporal ones, which is due to the larger number of observations in space than in time.

\begin{center}
\captionsetup{type=table}
\subfloat[]{\begin{tabular}{c|c|c|c}
& MEAN & RMSE & MAE \\
\hline
$\theta_1$ & 0.4033 & 0.0678 & 0.0559 \\
$\alpha_1$ & 1.4984 & 0.0521 & 0.0400 \\
$\theta_2$ & 0.2249 & 0.0649 & 0.0526 \\
$\alpha_2$ & 0.9563 & 0.0939 & 0.0767 
\end{tabular}}
\subfloat[]{\hspace*{0.5cm}\begin{tabular}{c|c|c|c|}
& MEAN & RMSE & MAE \\
\hline
$\theta_1$ & 0.4008 & 0.0668 & 0.0552 \\
$\alpha_1$ & 1.4946 & 0.0525 & 0.0400 \\
$\theta_2$ & 0.2188 & 0.0597 & 0.0489 \\
$\alpha_2$ & 0.9275 & 0.0976 & 0.0799 
\end{tabular} }
\captionof{table}{Mean, RMSE and MAE of the WLSEs when applied to exact realizations from the Brown-Resnick process (a) and to realizations with observational noise (b).}
\label{summarysemi}
\end{center}

\begin{figure} 
\centering
\subfloat[]{\includegraphics[scale=0.29]{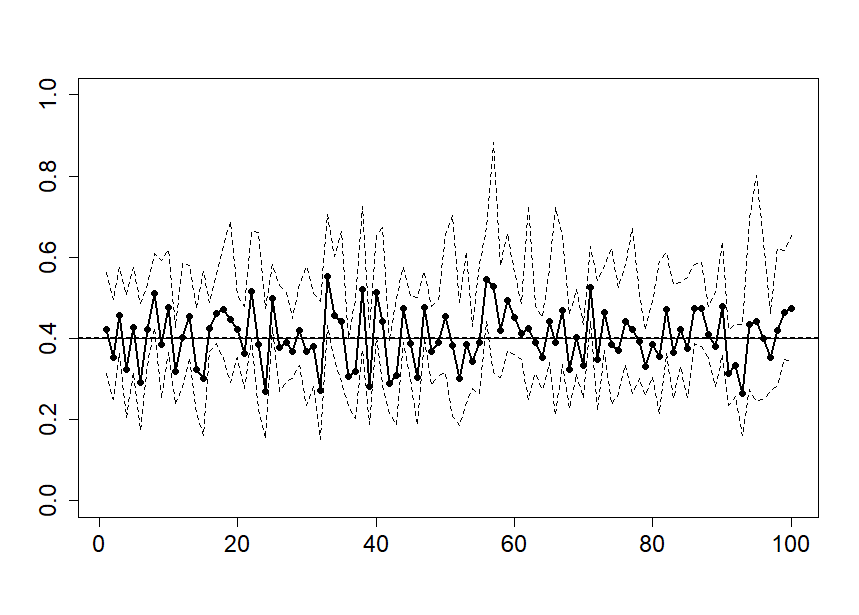}}
\subfloat[]{\includegraphics[scale=0.29]{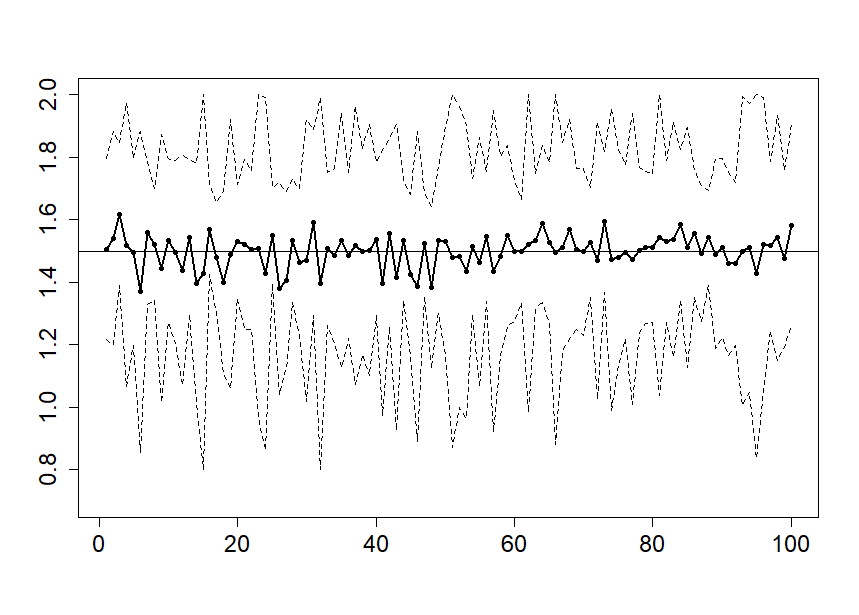}}\\[-7mm]
\caption{ WLSEs of $\theta_1$ (left) and $\alpha_1$ (right) for 100 simulated Brown-Resnick space-time processes together with pointwise $95\%-$subsampling confidence intervals (dashed). The middle solid line is the true value and the middle dashed line represents the mean over all estimates.} \label{Fig1}
\end{figure}

\begin{figure}
\centering
\subfloat[]{\includegraphics[scale=0.29]{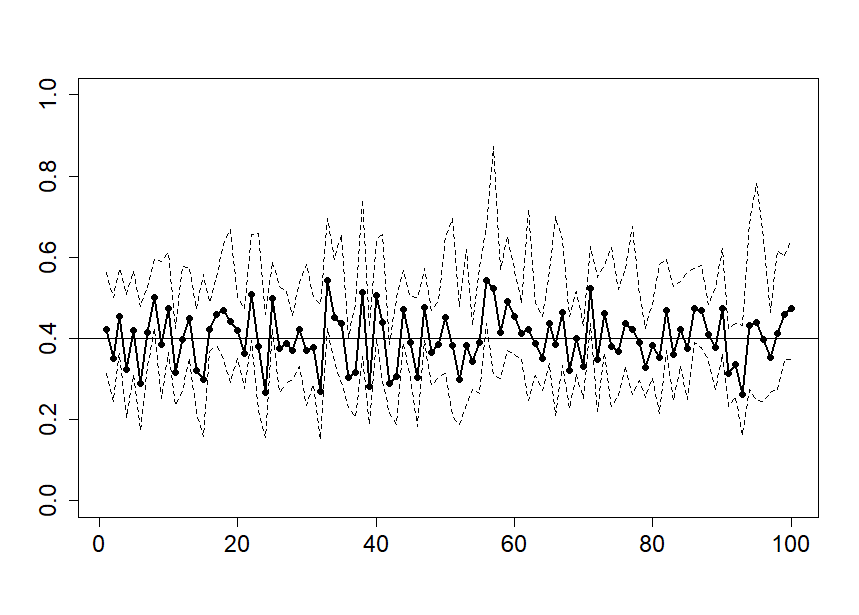}}
\subfloat[]{\includegraphics[scale=0.29]{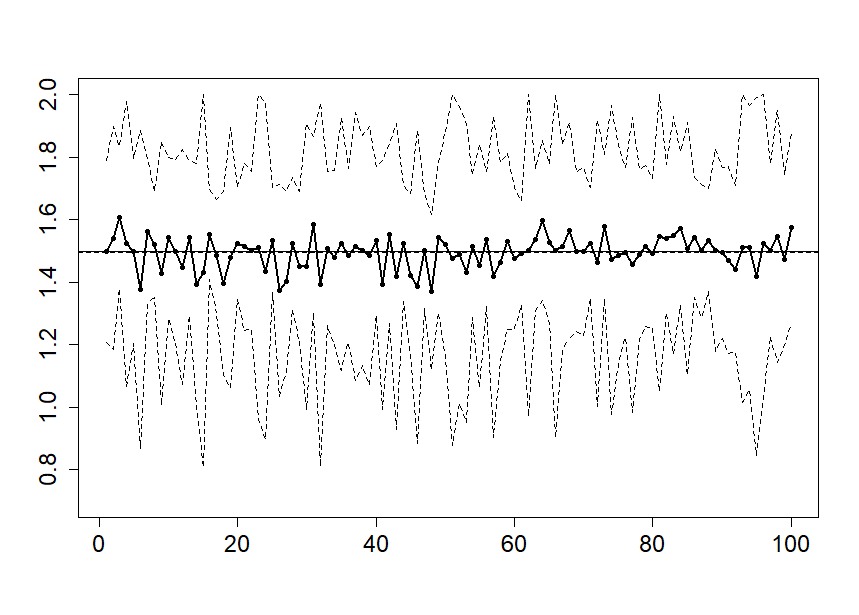}}\\[-7mm]
\caption{WLSEs of $\theta_1$ (left) and $\alpha_1$ (right) for 100 simulated Brown-Resnick space-time processes with noise together with pointwise $95\%-$subsampling confidence intervals (dashed). The middle solid line is the true value and the middle dashed line represents the mean over all estimates.} \label{space_pert}
\end{figure}

\begin{figure}
\centering
\subfloat[]{\includegraphics[scale=0.29]{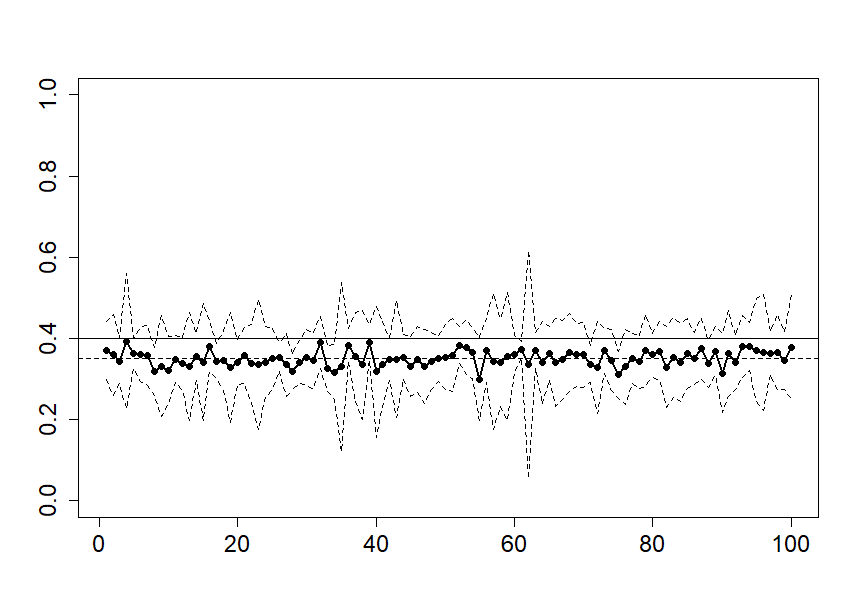}}
\subfloat[]{\includegraphics[scale=0.29]{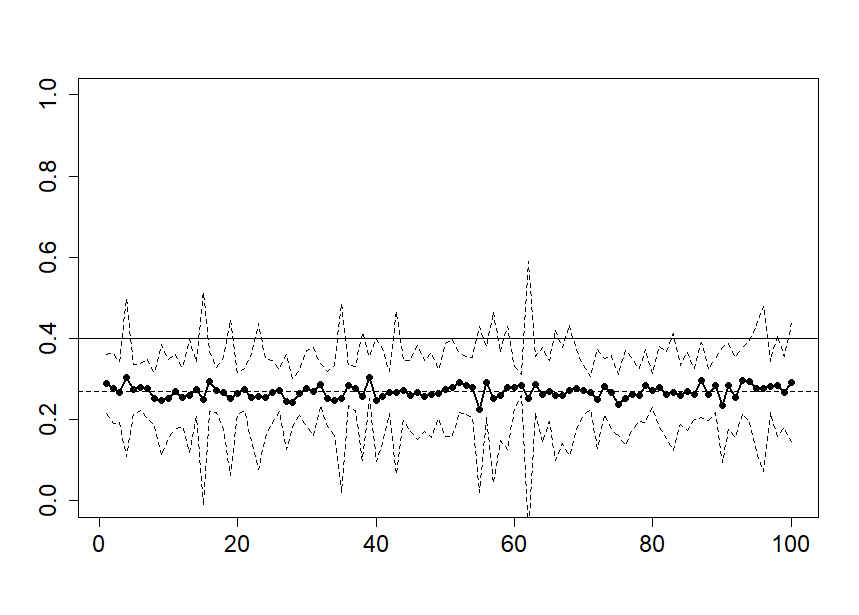}}\\[-7mm]
\caption{Pairwise likelihood estimates of $\theta_1$ for 100 simulated Brown-Resnick space-time processes together with pointwise $95\%-$subsampling confidence intervals (dashed) in the left-hand plot; the corresponding estimates and pointwise $95\%-$subsampling confidence intervals (dashed) for the  simulated processes with noise are presented in the right-hand plot.
The middle solid line is the true value and the middle dashed line represents the mean over all estimates.} \label{space_PLE_unper_pert}
\end{figure}

In comparison with pairwise likelihood estimation in finite samples (cf. \citet{Steinkohl2}), a big advantage of the semiparametric method is the substantial reduction of computation time by about a factor 15.
Moreover, the semiparametric estimation method is much more robust when applied to observations that reveal slight deviations from the model assumptions. 
To illustrate this point, we repeated the simulation study described above with data obtained from the original ones by adding to each measurement the absolute value of an independent $\mathcal{N}(0,0.2)$-distributed error. 
The results of the semiparametric estimation remain practically unaffected; see the summary measures in the panel (b) of Table~\ref{summarysemi} and Figure~\ref{space_pert}. This result can be explained theoretically by Lemma~\ref{le3.2ii}. Adding noise to observations of the Brown-Resnick process does not affect the underlying true extremogram nor the rate of convergence of the empirical extremogram to the true one. 
In contrast, when applying pairwise likelihood estimation to the same simulated data we observe that the estimates are much more sensitive to small disturbations than the semiparametric estimation.
Whereas for the original data the estimates are slightly biased, for the corrupted data the bias increases considerably; their variances, however, remain nearly unaffected and are (not surprisingly) smaller than the corresponding variances of the semiparametric estimates.
Figure~\ref{space_PLE_unper_pert} illustrates the pairwise likelihood estimates of the spatial parameter $\theta_1$ for the simulated Brown-Resnick space-time processes together with 95\%-subsampling confidence intervals and for the data with noise.

\bibliographystyle{plainnat}
\bibliography{bibtex_spacetime}